\documentclass[11pt]{article}
\pdfoutput=1

\usepackage{hyperref}
\usepackage{microtype}
\usepackage{xspace}
\usepackage{mathtools}
\usepackage[utf8]{inputenc}
\usepackage[T1]{fontenc}
\usepackage{color}
\usepackage{amsmath,amssymb,amsthm}
\usepackage[top=1in, bottom=1in, left=1in, right=1in]{geometry}
\usepackage{graphicx}
\graphicspath{{figures/}}
\usepackage{authblk}

\newcommand{\psat}{\textsc{Planar3SAT}\xspace}
\newcommand{\Delivery}{\textsc{Delivery}\xspace}

\newcommand{\TDelivery}{\textsc{Fast\-Delivery}\xspace}

\newcommand{\TWDelivery}{\textsc{Fast\-Efficient\-Delivery}\xspace}
\newcommand{\CDelivery}{\textsc{Combined\-Delivery}\xspace}

\newcommand{\pw}{\omega}
\newcommand{\pv}{\upsilon}
\newcommand{\E}{\mathcal{E}}
\newcommand{\T}{\mathcal{T}}

\newcommand{\NP}{$\mathrm{NP}$}

\newcommand{\APSP}{\mathrm{APSP}}

\newcommand{\bigO}{\mathcal{O}}

\DeclareMathOperator*{\argmin}{arg\,min}
\DeclareMathOperator*{\argmax}{arg\,max}
\providecommand{\prel}{\sim_{\mathit{P}}}

\theoremstyle{plain}
\newtheorem{theorem}{Theorem}
\newtheorem{lemma}[theorem]{Lemma}
\newtheorem{corollary}[theorem]{Corollary}

\newtheoremstyle{app}
  {}
  {}
  {\itshape}
  {}
  {\bfseries}
  {}
  { }
  {\thmname{#1}\thmnote{ #3}.}
\theoremstyle{app}
\newtheorem*{theorem*}{Theorem}
\newtheorem*{lemma*}{Lemma}

\begin{document}

\title{Collective fast delivery by energy-efficient agents\footnote{An extended abstract of this paper appeared at MFCS 2018~\cite{mfcs18}.
It erroneously claimed the single agent approach for variant (iii) to have approximation ratio 2.}}

\author[1]{Andreas Bärtschi}
\author[1]{Daniel Graf}
\author[2]{Mat\'{u}\v{s} Mihal\'{a}k}
\affil[1]{
	ETH Zürich, Department of Computer Science, Switzerland\\
	\texttt{$\left\{\right.$andreas.baertschi,\,daniel.graf$\left.\right\}$@inf.ethz.ch}
}
\affil[2]{
	Department of Data Science and Knowledge Engineering, Maastricht University, Netherlands\\
	\texttt{matus.mihalak@maastrichtuniversity.nl}
}

\maketitle

\begin{abstract}
	We consider $k$ mobile agents initially located at distinct nodes of an
	undirected graph (on $n$ nodes, with edge lengths).
	The agents have to deliver a single item from a given source node $s$ to a given target
	node $t$.  
	The agents can move along the edges of the graph, starting at time $0$,
	with respect to the following: 
	Each agent $i$ has a \emph{weight}~$\pw_i$ that defines the rate of
	energy consumption while travelling a distance in the graph, and a
	\emph{velocity}~$\pv_i$ with which it can move.
		
	We are interested in schedules (operating the $k$ agents) that result in
	a small \emph{delivery time} $\T$ (time when the item arrives at
	$t$), and small \emph{total energy consumption} $\E$.
	Concretely, we ask for a schedule that: either (i) Minimizes $\T$, (ii)
	Minimizes lexicographically $(\T,\E)$ (prioritizing fast delivery), or
	(iii) Minimizes  $\epsilon\cdot \T + (1-\epsilon)\cdot \E$, for a given
	$\epsilon \in \left( 0,1 \right)$.
		
	We show that $(i)$ is solvable in polynomial time, and show that (ii) is
	polynomial-time solvable for uniform velocities and solvable in time
	$\bigO(n+k\log k)$ for arbitrary velocities on paths, but in general is
	NP-hard even on planar graphs. As a corollary of our hardness result,
	(iii) is \NP-hard, too. We show that there is a $3$-approximation
	algorithm for (iii) using a single agent.
\end{abstract}

\section{Introduction}
\label{sec:introduction}

Technological development has allowed for low-cost mass production of small and simple mobile robots. 
Autonomous vacuum cleaners, mowers, or drones are some of the best known examples.
There are attempts to deploy such autonomous \emph{agents} to deliver physical goods -- \emph{packages} \cite{PostRobots2,USAToday}.
In the future, for delivering over longer distances, a \emph{swarm} of such autonomous agents is a likely option to be adapted, since the energy supply of the agents is limited, or the agents are simply required to operate locally, or simply because the usage of some agents is more costly than others.
A careful cooperation and planning of the agents is thus necessary to provide energy, time, and cost efficient delivery.
This leads to plentiful optimization problems regarding the operation of the agents.

Here we consider the problem of delivering a single package as quickly as possible from a source node $s$ to a target node $t$ in a graph $G=(V,E)$ with edge lengths by a team of $k$ agents.
The agents have individual velocities, with which they can move along the edges of the graph, and also an energy-consumption rate for a travelled unit distance. 
The goal is to design \emph{centralized algorithms} to coordinate the agents such that the package is delivered from $s$ to $t$ in an efficient way. 
In the literature, delivery problems focusing solely on energy efficiency have been studied. 
One research direction considers every agent to have an initial amount of energy (battery) that restricts the agents' movements \cite{AnayaCCLPV16,DDalgosensors13}. 
The decision problem of whether the agents can deliver the package has been shown to be strongly \NP-hard on planar graphs~\cite{sirocco16,TCS17} and weakly \NP-hard on paths~\cite{DDicalp14}, and it remains \NP-hard on general graphs even if the agents can exchange energy~\cite{EnergyExchange15}.
The second research direction considers every agent to have unlimited energy supply, and an individual energy-consumption rate per travelled distance~\cite{STACS17,ATMOS17}. 
The problem of delivering the package and minimizing the total energy consumption can be solved in time $\bigO(k + n^3)$~\cite{STACS17}. 

In this paper, we primarily focus on delivering the package in a quickest possible way, and only secondarily on the total energy that is consumed by the agents. This has not been, to the best of our knowledge, studied before. Specifically, we consider the algorithmic problem of finding a delivery schedule that: (i) minimizes the delivery time, (ii) minimizes the delivery time using the least amount of energy, and (iii) minimizes a linear combination of delivery time and energy consumption.

\paragraph{Our model.} 
We are given an undirected graph $G=(V,E)$ on $n=|V|$ nodes. Each edge $e \in E$ has a positive \emph{length} $l_e$. The length of a path is the sum of the lengths of its edges. We consider every edge $e=\{u,v\}$ to consist of infinitely many points, where every point is uniquely characterized by its distance from $u$, which is between $0$ and $l_e$. We consider every such point to subdivide the edge $\{u,v\}$ into two edges of lengths proportional to the position of the point on the edge.
The distance $d_G(p,q)$ between two points $p$ and $q$ (nodes or points inside edges) of the graph is the length of a shortest path from $p$ to $q$ in $G$.
There are $k$ mobile agents initially placed on nodes $p_1,\ldots,p_k$ of $G$. Every agent $i=1,\ldots,k$ has a \emph{weight} $0\leq \pw_i < \infty$ and a \emph{velocity} $0 < \pv_i \leq \infty$.
Agents can traverse the edges of the graph. To traverse an edge $e$ (in either direction), agent $i$ needs time $l_e/\pv_i$ and $\pw_i\cdot l_e$ units of energy.

Furthermore there is a single package, initially (at time 0) placed on a source node $s$, 
which has to be delivered to a given target node $t$. 
Each agent can walk from its current location to the current location of the package (along a path in the graph), 
pick the package up, carry it to another location (a point of the graph), and drop it there.
From this moment, another agent can pick up the package again.
Only the moving in the graph takes time -- picking up the package and dropping it off is done instantaneously.
(The time spent by the package being dropped at a point until picked up again is, however, taken into account.)

We call a schedule that operates the agents such that the package is delivered a \emph{solution}. In such a schedule $S$, we denote by 
$d_i(S)$ the total distance travelled by agent $i$, and by $d_i^*(S)$ the distance travelled by agent $i$ while carrying the package.
The total energy consumption of the solution is thus $\E(S) = \sum_{i=1}^k \pw_i\cdot d_i(S)$ and 
the time needed to deliver the package is given by $\T(S)=\sum_{i=1}^k d_i^*(S)/\pv_i + (\text{the overall time the package is not carried}).$
\emph{Fast and energy-efficient} \Delivery is the optimization problem of finding a solution that has small delivery time $\T$ as well as total energy consumption $\E$.
We study the following three objectives (see Figure~\ref{fig:delivery} for illustration):
\vspace{-1ex}
\begin{description}
	\setlength{\itemsep}{0ex}
	\item[(i)]	Minimize the delivery time $\T$.
	\item[(ii)]	Lexicographically minimize the tuple $(\T,\E)$, i.e.\ among all solutions with minimum $\T$,\newline
		\hspace*{-1ex}find a solution that has minimum energy consumption $\E$.
	\item[(iii)]	Minimize a convex combination $\epsilon\cdot \T + (1-\epsilon)\cdot \E$, for some given value $\epsilon \in \left( 0,1 \right)$.
\end{description}
Recent parallel work studied the following complementary -- energy focused -- variants:
\vspace{-1ex}
\begin{description}
	\setlength{\itemsep}{0ex}		
	\item[\textnormal{(iv)}]	Lexicographically minimize the tuple $(\E,\T)$, i.e.\ prioritize the minimization of $\E$~\cite{FCT17}.
	\item[\textnormal{(v)}]	Minimize the energy consumption $\E$~\cite{STACS17,ATMOS17}.
\end{description}

\begin{figure}[t!]
	\includegraphics[width=\linewidth]{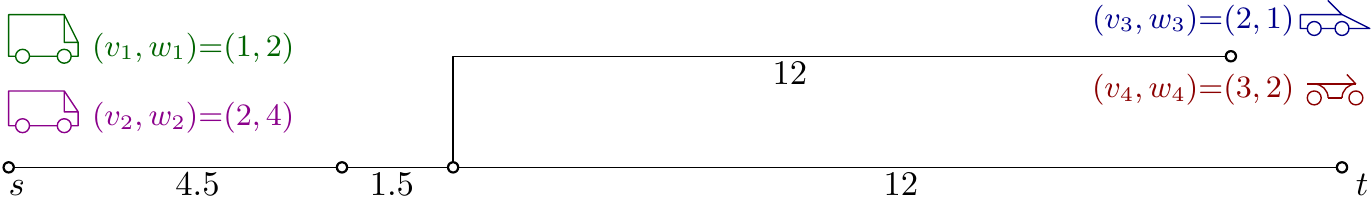}
	\caption{Example for optima of variants of fast \emph{and} energy-efficient \Delivery:\newline
	(ii) Using agents $2$ and $4$, we get 
	$(\T,\E) = (\max\left\{6/2,12/3\right\} + 12/3,\ 4\cdot 6+2\cdot 12 + 2\cdot 12) = (8,72)$.\newline
	(iii) For $\epsilon=\tfrac{4}{5}$, using agents $1$ and $4$, we get 
	$\tfrac{4}{5}\T = \tfrac{4}{5}\left(\max\left\{ 4.5/1, (12+1.5)/3 \right\}+(1.5+12)/3\right)$ 
	and $\tfrac{1}{5}\E = \tfrac{1}{5}(4.5\cdot 2+(12+1.5)\cdot 2+(1.5+12)\cdot 2)$ for a combined total of 
	$\tfrac{4}{5}\left(4.5+4.5\right)+\tfrac{1}{5}(9+27+27) = 19.8$.\newline
	(iv) Using agents $1$ and $3$, we get 
	$(\E,\T) = (2\cdot 6+1\cdot 12 + 1\cdot 12,\ \max\left\{6/1,12/2\right\} + 12/2) = (36,12)$.\label{fig:delivery}}
\end{figure}

\noindent In all variants it is natural to (without loss of generality) only consider simple paths as the trajectory of the package, i.e., if at times $t_1, t_2$ ($0 \leq t_1 \leq t_2 \leq \mathcal{T}$) the package is at the same position $p$, then it remains at position $p$ for the time in-between ($\forall t\in \left[ t_1, t_2 \right]$). 
We will make this assumption throughout this paper.

\paragraph{Our contribution.}
First, in Section~\ref{sec:preliminaries}, we prove for the first time that optimum solutions exist for all mentioned variants of \Delivery
(while previous work on (iv) and (v) implicitly assumed this).  
Then, in Section~\ref{sec:time-only}, we investigate the problem of minimizing the delivery time $\T$ only.
We call this optimization problem \TDelivery and show that there is a polynomial-time 
dynamic program of time complexity $\bigO(k^2|E|+ k|V|^2 + \APSP) \subseteq \bigO(k^2n^2 + n^3)$, where $\bigO(\APSP)$ is the running time of an all-pair shortest path algorithm for undirected graphs. 

In Section~\ref{sec:time-first}, we study \TWDelivery, prioritizing the delivery time $\T$ over the energy consumption $\E$. 
We first show that the problem can be solved in polynomial time for uniform velocities. However, we prove the problem to be \NP-hard for general velocities even on planar graphs. We therefore consider the restricted graph class of \emph{paths}, in which we can decompose the problem into uniform velocity instances. 
For each such instance, we establish a characterization of handover points. 
Using geometric point-line duality~\cite{edelsbrunner1987algorithms} and dynamic planar convex hull techniques~\cite{RikoConvexHull02},
we give an $\bigO(n+k\log k)$ algorithm for paths.

In Section~\ref{sec:combination}, we show that for arbitrary given weights $\epsilon \in (0,1)$, 
the minimum convex combination $\epsilon\cdot \T+ (1-\epsilon)\cdot \E$ can be $3$-approximated by a single agent,
while \NP-hardness follows from an adaptation of the hardness proof in the preceding section.
We call the task of minimizing the convex combination \CDelivery.
Finally, in Section~\ref{sec:discussion} we discuss several extended models to which our approach can be generalized.
Some technical proofs are omitted in the main text and instead provided in the appendix, 
or also as inlined proofs in a thesis on several variants of \Delivery~\cite{BaertschiPhD}.

\paragraph{Comparison to related work.} Among the earliest problems related to \Delivery are the Chinese Postman Problem~\cite{Edmonds73} and the Traveling Salesman Problem~\cite{ApplegateTSP}, 
in which a single agent has to visit multiple destinations located in edges or nodes of the graph, respectively. 
The latter has given rise to a class of problems known as Vehicle Routing Problems~\cite{TothVRP}, which are concerned with the distribution of goods by a fleet of (homogeneous) vehicles 
under additional hard constraints such as time windows.  
Minimizing the total or the maximum travel distance of a group of agents for several tasks such as the formation of configurations~\cite{Demaine2009} or the visit of designated arcs~\cite{Frederickson76}
have been studied for identical agents as well. 
Energy-efficient \Delivery (without optimization of delivery time) has been recently introduced~\cite{STACS17} for an arbitrary number of packages,
with handovers restricted to take place at nodes of the graph only. 
This setting turns out to be \NP-hard, but can be solved in polynomial-time for a single package, 
in which case the restriction of handovers to nodes becomes irrelevant (there is always an optimal solution which does not use any in-edge handovers). 
To the best of our knowledge, this present paper and a parallel work~\cite{FCT17} on variant (iv) 
are the only ones studying the \Delivery problem with agents which have different velocities.
Similar to our approach, the latter studies a uniform weight setting first. The uniform weight result is then used as a subroutine in a dynamic program for general weights. Our hardness result shows that such an approach (combination of uniform velocities) 
is not possible for \TWDelivery, even on planar graphs.
Finally, mobile agents with distinct maximal velocities have been getting attention in areas such as searching~\cite{BampasSirocco16}, walking~\cite{Czyzowicz15} and patrolling~\cite{Czyzowicz11}.

\section{Preliminaries}
\label{sec:preliminaries}

We first formally establish that optimum solutions for all variants of efficient \Delivery exist. 
To this end, we assume without loss of generality that each agent carries the package at most once.
(If an agent carries the package twice, we replace all agents acting in-between this agent by the agent itself. 
By the triangle inequality, neither the delivery time nor the energy consumption increases.)
Each solution which operates agents $i_1,i_2,\ldots,i_{\ell}$ in this order can be represented by the
drop-off locations of these agents only (note that for two consecutive agents $i,j$, 
the drop-off location of agent~$i$, denoted by $q_i^-$, corresponds to the pick-up location of agent $j$, denoted by $q_j^+$). 
Since we allow in-edge handovers, there are infinitely many solutions -- however, these can be divided into 
\emph{finitely} many topologically compact sets.
As $\E,\T$ act as continuous functions on these sets, we have in each set a minimum solution. 

\begin{theorem}[Existence of optimum solutions]
	There exists an optimum solution minimizing the delivery time $\T$ 
	(the energy consumption $\E$, or $\epsilon\cdot \T +(1-\epsilon)\cdot \E$, $(\T,\E)$, $(\E,\T)$, respectively).%
	\label{thm:existence}
\end{theorem}

\section{Optimizing delivery time only}
\label{sec:time-only}

Throughout this section, we assume that all agents have weight $\pw_i = 0$. 
Hence in all three variants of fast energy-efficient \Delivery, $\E=0$ 
and we are after a solution for delivery with earliest-possible delivery time. 
We show that \TDelivery is polynomial-time solvable, due to the following characterization 
of optimum solutions (which exist by Theorem~\ref{thm:existence}): 

\begin{lemma}
	For every instance of \TDelivery, there is an optimum solution in which
	(i) the velocities of the involved agents are strictly increasing, 
	(ii) no involved agent arrives at its pick-up location earlier than the package (carried by the preceding agent), and
	(iii) if more than one agent is involved in transporting the package over an edge $\left\{ u,v \right\}$ in direction from $u$ to $v$, 
	then only the first involved agent will ever visit $u$.
	\label{lem:increasing-speeds}
\end{lemma}

\begin{proof}
	All three properties can be shown by exchange arguments. Taking any optimum solution,
	we turn it into an optimum solution that adheres to the three properties as follows:
	
	(i) Label the agents $1,2,\ldots,i,\ldots$ in the order in which they transport the package. 
	Let $i$ be the first agent such that $\pv_i \geq \pv_{i+1}$.
	Now we can simply replace agent $i+1$ by letting agent $i$ travel on the same trajectory 
	on which $i+1$ transported the package; and by doing so, we don't increase the delivery time.
	
	(ii) Let $i$ be the first agent that has to wait at its pick-up location for the package to arrive. 
	Instead of waiting, we let $i$ proceed on the original trajectory of the package towards $s$ 
	until it meets the preceding agent $i-1$. Handing over the package at this new spot 
	cannot increase the delivery time $\T$, as $\pv_{i-1} < \pv_i$ (we only increase velocities along the trajectory). 
	However, $\T$ might remain constant if this increase in velocity is countered by a longer waiting time of the package 
	at the handover to agent $i+1$.

	(iii) Assume that multiple agents bring the package from $u$ to $v$ over the edge $\left\{ u,v \right\}$, by visiting $u$ first. 
	By assumption (i) the last such agent $i$ has the highest velocity and thus agent $i$ can just as well pick up the package at 
	$u$ without the help of the other agents.
\end{proof}


\begin{corollary}
	After a preprocessing step of time $\bigO(k+|V|)$ -- in which we remove in each node 
	all but the agent with maximum velocity $\pv_i$ --
	we may assume that $k\leq |V|$.
	\label{cor:time:ksmallerV}
\end{corollary}

\paragraph{Towards a dynamic program.}
Making use of characterization (i) of Lemma~\ref{lem:increasing-speeds}, 
we relabel the agents such that $\pv_1\leq \pv_2\leq \ldots\leq \pv_k$. 
We can then look at subproblems where we only use the first $i-1$ among all $k$ agents. 
Assume node $v^*$ is the first node that the new agent $i$ (starting at $p_{i}$) passes while actually carrying the package. 
According to characterizations (ii) and (iii), when defining the recursion, 
we have to take care of these two cases, see Figure~\ref{fig:cases-a-b}:
\begin{description}
	\item[a)]	Agent $i$ might arrive at node $v^*$ `late', 
			the package has already been dropped off there before by one of the agents $1,2,\ldots,i-1$ and had been waiting.
	\item[b)]	Agent $i$ might arrive at node $v^*$ `early', 
			in which case it should walk towards the package to receive it earlier 
			and bring it back to $v^*$ faster (having larger velocity than the currently carrying agent, after all). 
			In this case, agent $i$ picks up the package at a point $p$ 
			which is strictly in the interior of the edge $\left\{ u,v^* \right\}$ and 
			which is as close to node $v^*$ as possible,
			i.e., $p$ must be reachable by both agent $i$ \emph{and} the package 
			-- carried by only the first $i-1$ agents -- at the \emph{earliest} possible time: 
			$(d(p_{i},v^*) + d(v^*,p))/\pv_{i}$.
\end{description}

\begin{figure}[t!]
	\centering
	\includegraphics[width=\linewidth]{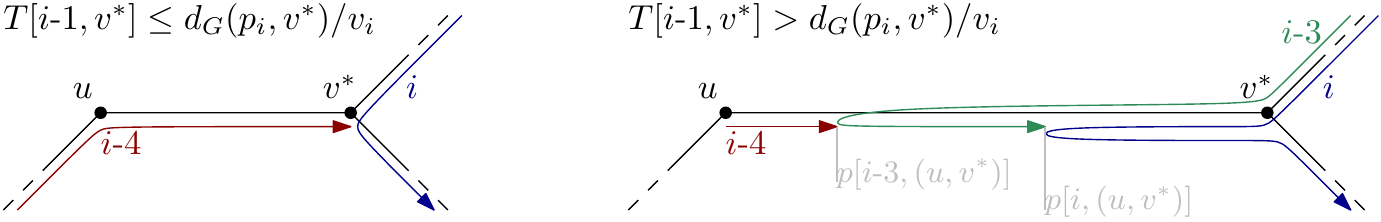}
	\caption{Examples for cases a) and b): (left) Agent $i$ picks up the package at node $v^*$.\newline 
	(right) Agent $i$ picks up the package inside the edge $(u,v^*)$ at the earliest possible time.}
	\label{fig:cases-a-b}
\end{figure}

\paragraph{Dynamic Program.}
First we are interested in the distance between any two nodes in the graph, 
which we can find with an \emph{all-pair-shortest-paths} algorithm $\APSP$. 
We denote the time needed for this precomputation by $\bigO(\APSP)$.
Then, given the agents in ascending order of their velocities $\pv_i$, 
for each prefix $1,2,\ldots i$ of the agent order and each node $v$ we define the following subproblem:
\begin{align*}
	S\left[ i,v \right]	=	& \text{ A fastest schedule to bring the package to node $v$ 
						 using agents $\left\{ 1,\ldots,i \right\}$.}\\
	\T\left[ i,v \right]	=	& \text{ The time needed in $S\left[ i,v \right]$ to deliver the package to $v$. }\\
	A\left[ i,v \right]	=	& \text{ Index of the last agent to carry the package in $S\left[ i, v \right]$.}\\	
	p\left[ i, (u,v) \right]=	& \text{ The pick-up point $p$ strictly inside edge $\left\{ u,v \right\}$ 
						 and closest to $v$, reachable }\\
					& \text{ by \emph{both} the package (coming from $u$, 
						 delivered by agents $1,\ldots,i-1$) \emph{and} }\\ 
					& \text{ agent $i$ (coming via $v$)
						 in time $(d(p_{i},v) + d(v,p))/\pv_{i}$ (if applicable).}
\end{align*}
Note that although our graph only has undirected edges, $p[i,(u,v)]$ considers an \emph{ordered} tuple of nodes $(u,v)$, 
denoting that the package is transported from $u$ to $v$. Thus $p[i,(v,u)]$ has the analogous meaning of the package 
crossing edge $\left\{ u,v \right\}$ from $v$ towards $u$. Both $p[i,(u,v)]$ and $p[i,(v,u)]$ 
might be undefined, as can be seen below.

We compute the optimum delivery times $\T[i,v]$ (together with $A[i,v]$) without explicitly maintaining the schedules $S[i,v]$. 
A concrete final schedule $S$ can then be retraced from $A[,]$, see Theorem~\ref{thm:speedy-delivery}. 
For computing $\T[i,v]$ and $A[i,v]$ we `guess' the first node $v^*$ of cases a) and b) above by trying each node $v$ as a candidate. 
We then can compute $\T[i,v]$ and $A[i,v]$ for all other nodes using the pre-computed distances between all pair of nodes: 
\begin{enumerate}
	\item	\emph{Initialization:}
		For all nodes $v$, we initialize $S[i,v] := S[i-1,v],A[i,v] := A[i-1,v]$ and $\T[i,v] := \T[i-1,v]$. 
		This automatically takes care of case a), where the package arrives at $v$ before agent $i$ can reach $v$.
	\item	\emph{In-edge pick-ups:}
		We go over all node pairs $(u,v)$ such that $\left\{ u,v \right\} \in E$ and check whether agent $i$ 
		can pick up the package inside $\left\{ u,v \right\}$ to advance it to node $v$ faster than in schedule $S[i-1,v]$. 
		We do so by checking whether we have $d(p_i,v)/\pv_i < \T[i-1,v]$ \emph{and} $d(p_i,u)/\pv_i > \T[i-1,u]$. 
		In this case, agent $i$ receives the package from a previous agent $j$ that brought it from $u$ or from $p[j,(u,v)]$. 
		Thus we get a set $P$ of candidates for $p[i,(u,v)] := \smash{\arg \min_{p \in P}} \left\{ d(p,v) \right\}$.
		The candidate set $P$ consists of all points $p$ strictly inside the edge $\left\{u,v \right\}$ 
		such that there exists an agent of index $j$, $A[i-1,u] \leq j < i$, for which we have
		\begin{align*}
			\max\left\{ \T[i-1,u], \frac{d(p_j,u)}{\pv_j} \right\} + \frac{d(u,p)}{\pv_j}
					& = \frac{d(p_i,v)\ +\ d(v,p)}{\pv_i}	\\
		\intertext{if $j$ is coming from $u$, \emph{or} -- if $p[j,(u,v)]$ is defined -- }
			\frac{d(p_j,v)\ +\ d(v,p[j,(u,v)])\ +\ d(p[j,(u,v)], p)}{\pv_j} 
					& = \frac{d(p_i,v)\ +\ d(v,p)}{\pv_i}.
		\end{align*}
		Having computed $p[i,(u,v)]$ as the point in $P$ closest to $v$, we update node $v$ accordingly: 
		Set $\T[i,v] := \min\left\{ \T[i,v], \tfrac{d(p_i,v) + 2d(p[i,(u,v)],v)}{\pv_i} \right\}$, where using `$\min$' takes 
		care of cases in which we have multiple incident edges to $v$ that all potentially have in-edge pick-ups by $i$, 
		and set $A[i,v] = i$ (valid since we consider the case where $d(p_i,v) < \T[i-1,v]$).
	\item	\emph{Updates:} So far we have computed the subproblems $S[i,v]$ correctly,
		if node $v$ corresponds to the first node $v^*$ of cases a) and b) 
		(in particular we checked whether the faster agent $i$ can help to advance the package over only one edge).
		Now we also consider all cases where agent $i$ transports the package over arbitrary distances, by updating
		all other schedules $S[i,u]$ accordingly:
		For each node $v$, for each node $u$, if $\T[i,u] > \max\left\{ \T[i,v], d(p_i,v)/\pv_i \right\} + d(v,u)/\pv_i$ we set 
		$A[i,u] := i$ and $\T[i,u] := \max\left\{ \T[i,v], d(p_i,v)/\pv_i \right\} + d(v,u)/\pv_i$.
\end{enumerate}

\begin{theorem}
	An optimum schedule for \TDelivery of a single package can be computed in time 
	$\bigO(k^2|E| + k|V|^2 + \APSP) \subseteq \bigO(k^2n^2 + n^3)$.
	\label{thm:speedy-delivery}
\end{theorem}
\begin{proof}
	For each $i$ from $1$ to $k$ we can compute all values $A[i,v],\ \T[i,v]$ in time $\bigO(|V|)$ for the initialization, 
	$\bigO(|E|k)$ to check for in-edge pick-ups and $\bigO(|V|^2)$ for the updates (for which we need precomputed all-pair shortest paths).
	Overall we get a running time of $\bigO(\APSP + k^2|E| + k|V|^2)$.
	The delivery time is then given in $\T[k,t]$. Correctness of the algorithm follows from the definition of the subproblems 
	and the case distinction stemming from Lemma~\ref{lem:increasing-speeds}.
	Since we did not explicitly maintain the schedules $S[i,v]$, we retrace the concrete schedule $S$ from $A[,]$ by backtracking: 
	Let $i$ denote the last used agent $A[k,t]$.
	We can find $i$'s `first node' $v^*$ in time $\bigO(|V|)$ by searching for the smallest value $\T[i,u]$ such that  
	\[	\max\left\{ \T[i,u], d(p_i,u)/\pv_i \right\} = \T[k,t]-d(u,t)/\pv_i.	\]
	If $A[i,v^*] \neq i$, we recurse, otherwise we find the correct adjacent node and all in-edge handovers 
	by looking -- for each of the $\bigO(\deg(v^*))$ many neighbors $u$ of $v^*$ -- at the overall $\bigO(k\deg(v^*))$ many values
	$p[j,(u,v^*)]$ (where $j\leq i$) and $\T[j,u]$ (where $j<i$). 
\end{proof}

\section{Prioritizing delivery time over energy consumption}
\label{sec:time-first}

In this Section, we want to find the most efficient among all fastest delivery schedules.
We call this problem \TWDelivery and will first show that it can be solved in polynomial time 
for uniform velocities ($\forall i,j\colon \pv_i = \pv_j$), due to a characterization of 
optimum schedules.
In contrast, we prove \NP-hardness for arbitrary speeds, even on planar graphs. 
However, for paths we show how one can subdivide general instances into phases of concecutive agents having the same velocity, 
and achieve an efficient $\bigO(n+k\log k)$-time algorithm.

\subsection{A polynomial-time algorithm for uniform velocities}
\label{sec:time-first-uniform}

\begin{lemma}
	Consider \TWDelivery on instances with uniform agent velocities and let 
	$\delta$ denote the offset of the closest agent's starting position to $s$. Then there exists an optimum
	schedule such that the pick-up position $q_i^+$ of each involved agent $i$ satisfies:
	\begin{itemize}
		\item	$d(s,q_i^+) + d(q_i^+,t) = d(s,t)$, i.e., $q_i^+$ lies on a shortest $s$-$t$-path, and
		\item	$d(p_i, q_i^+) \leq \delta + d(s, q_i^+)$, with equality if $q_i^+$ lies strictly inside an edge.
	\end{itemize}
	\label{lem:uniform-characterization}
\end{lemma}

\begin{proof}
	Since all agents have the same velocity $\overline{\pv}$, any fastest delivery of the package must follow a shortest path from $s$ to $t$.
	Furthermore, since the closest agent could deliver the package on its own in time $(\delta + d(s,t))/\overline{\pv}$, 
	each involved agent $i$ has to arrive at $q_i^+$ no later than the package itself, giving $d(p_i, q_i^+) \leq \delta + d(s, q_i^+)$.
	It remains to show that we can modify every optimum solution into an optimum solution in which we have 
	$d(p_i, q_i^+) = \delta + d(s, q_i^+)$ whenever $q_i^+$ lies strictly inside an edge $e=\{u,v\}$.
	Denote by $i$ the first agent for which this is not the case and by $i-1$ its preceding agent. 
	Assume that the package enters $e$ via $u$ (i.e.~$d(s,u)<d(s,v)$). 
	Note that $i$ must have entered $e$ via $v$, 
	since otherwise the energy consumption could be improved by letting $i$ pick up the package already at $u$ 
	(without increasing the delivery time), contradicting the optimality of our solution.
	Now we distinguish two cases relating the weights $\pw_i$ and $\pw_{i-1}$, 
	yielding either a decrease of the energy consumption, 
	or a possibility to move $q_i^+$ to a position satisfying the characterization.
	\begin{itemize}
		\item	$2\pw_i > \pw_{i-1}$: Moving $q_i^+$ by $\epsilon>0$ towards $v$ decreases $\E$ by 
			an amount of $(2\pw_i-\pw_{i-1})\cdot \epsilon>0$.
		\item	$2\pw_i \leq \pw_{i-1}$: We move $q_i^+$ towards $u$ (without increasing neither delivery time nor 
			energy consumption) until we reach $q_i^+ = u$, or $q_i^+$ inside the edge $\left\{ u,v \right\}$ 
			such that $d(p_i, q_i^+) = \delta + d(s, q_i^+)$,
			or $q_i^+ = q_{i-1}^+$. In the last case, discarding agent $i-1$ from our solution 
			results in an energy consumption decrease of at least $\pw_{i-1}\cdot d(p_{i-1},q_{i-1}^+)>0$. \qedhere
	\end{itemize}	
\end{proof}

\paragraph{Polynomial-time algorithm}
We use the characterization in Lemma~\ref{lem:uniform-characterization} to find an optimum solution for \TWDelivery
of delivery time $\T = (\delta + d(s,t))/\overline{\pv}$: 
For each agent~$i$, we compute the set $Q_i$ of all potential pick-up locations, i.e.,
the set of points $q_i$ that satisfy Lemma~\ref{lem:uniform-characterization}. 
The number of potential locations is $|Q_i| \in \bigO(|V|+|E|) \subseteq \bigO(n^2)$. 
Then we build an auxiliary directed acyclic multi-graph $H$ on a node set $V(H) = \bigcup_{i=1}^k Q_i$, 
of size $|V(H)| \in \bigO(|V|+k|E|) \subseteq \bigO(kn^2)$. 
Each directed edge in $E(H)$ describes how agent $i$ can contribute to the delivery by bringing the package from its starting position $q_i$ 
to another agent's starting position $q_j$ along a shortest $s$-$t$-path: 
For each pair of nodes $q_i \in Q_i$ and $q_j \in V(H)$ such that $q_i \neq q_j$ and 
$d(s,q_i)+d(q_i,q_j)+d(q_j,t) = d(s,t)$, we add an arc $(q_i,q_j)$ of weight $\pw_i\cdot (d(p_i,q_i)+d(q_i,q_j))$ to $E(H)$.
Overall, we have at most $|E(H)| \in \bigO(k\cdot n^2 \cdot kn^2)$ many arcs.
By construction of $H$, running Dijkstra's shortest path algorithm on the multi-graph $H$ 
finds a shortest path from $s$ to $t$ corresponding to an optimal solution.
\begin{theorem}
	An optimum solution for \TWDelivery can be found in time $\bigO(k^2n^4)$,
	assuming all agents have the same velocity.
	\label{thm:uniform-velocities}
\end{theorem}

\subsection{\NP-hardness on planar graphs}

Contrary to \TDelivery (where we had non-decreasing velocities $\pv_i$), 
when prioritizing delivery time \emph{but still regarding energy consumption}, 
we can't characterize the order of the agents by their coefficients $(\pv_i, \pw_i)$:
Consider an instance in which both the starting position $p_a$ of the absolutely fastest agent $a$ 
as well as the package destination $t$ are separated from the rest of the graph by two very long edges $q_a^+$---$p_a$, $q_a^+$---$t$. 
Then in \emph{every} fastest solution, agent~$a$ (with $\pv_a$ large, e.g.~$8$)
must deliver the package from $q_a^+$ to $t$, see Figure~\ref{fig:reduction-idea}~(right).

In \TWDelivery, the task is thus to balance slow but efficient agents (with, e.g., $\pv=1,\pw=0$)
and fast inefficient agents (with, e.g., $\pv=2,\pw=1$) 
to collectively deliver the package to $a$'s pick-up location $q_a^+$ 
\emph{just-in-time} -- i.e., in time $d(p_a,q_a^+)/\pv_a$ -- \emph{without using too much energy}.
We can construct suitable instances by a reduction from \psat~\cite{planar3sat82} \emph{(Sketch)}: Starting from a planar formula $F$
in three-conjunctive normal form, as in Figure~\ref{fig:reduction-idea}~(left), we build a delivery graph $G(F)$.
This can be done such that the instance is guaranteed to have schedules with minimum delivery time, i.e.~with $\T = d(p_a,t)/\pv_a$. 
However, there should only be such a minimum-time schedule which simultaneously has low energy consumption $\E$
\emph{if and only if} the formula $F$ has a satisfiable variable assignment.

To this end, we place the fast agents on nodes corresponding to variables and literals. Intuitively, these agents decide on the 
routing of the package, thus setting the assignment of each variable.
The slow agents, on the other hand, are placed on clause nodes, each clause receiving just one agent short of the number of its literals.
Intuitively, for a just-in-time delivery to $q_a^+$ with small energy consumption, 
each clause has to spend one of its agents for each of its unsatisfied literals. 
By construction, this is only possible if each clause is satisfied:

\begin{theorem}
	\TWDelivery is \NP-hard, even on planar graphs.	
	\label{thm:fast-efficient-hardness}
\end{theorem}

\begin{figure}[t!]
	\centering
	\includegraphics[width=\linewidth]{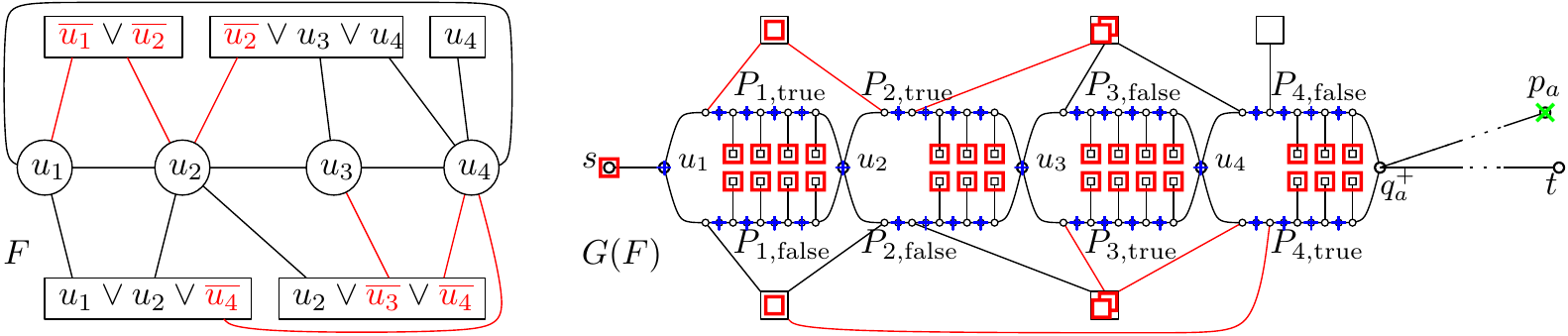}
	\caption{(left) A planar 3CNF formula $F$, satisfiable by $(u_1,u_2,u_3,u_4) = (\text{true, false, false, true}).$ 
	\hfill (right) Its transformation into a corresponding delivery graph $G(F)$. 
	The satisfiable assignment~of~$F$ corresponds to a low-cost delivery in $G(F)$ via paths 
	$P_{1,\mathrm{true}}, P_{2,\mathrm{false}}, P_{3,\mathrm{false}}, P_{4,\mathrm{true}}$, and vice versa.
	We have slow agents for clauses (\textcolor{red}{$\square$}), fast agents for variables/literals (\textcolor{blue}{$+$})
	and a very fast agent (\textcolor{green}{$\times$}).}
	\label{fig:reduction-idea}
\end{figure}

\subsection{An efficient algorithm for paths}

The preceding hardness result raises the question for which restricted graph classes we can expect an efficient algorithm 
for arbitrary velocity instances. To contribute to this question it is natural to study paths -- on paths, 
the V-shaped $p_a$---$q_a$---$t$ component attached to the rest of the graph, as used in the hardness proof, `collapses' to a line.
We show that this allows us to decompose the problem into linearly many uniform velocity instances in time $\bigO(n+k\log k)$.  
Theorem~\ref{thm:uniform-velocities} then implies that \TWDelivery can be solved in polynomial-time.
Improving on this by a careful analysis of paths, we show how to solve each uniform velocity instance 
in time $\bigO(n+k\log k)$ as well, and that these instances can be combined in time $\bigO(k)$, 
giving an overall $\bigO(n+k\log k)$-time algorithm.

\paragraph*{Decomposition into uniform velocity instances}
In the following, we look at the path graph $G$ as the real line, and assume 
(after performing a depth-first search from $s$ and ordering the starting positions in time $\bigO(n+k\log k)$) 
without loss of generality that $s=0<t$, that $p_1 \leq p_2 \leq \ldots \leq p_k$ and that $n=k+2$, 
as the only relevant nodes on the line are $s$, $t$ and the starting positions $p_i$.
Note that in an optimum solution of \TWDelivery, no agent $i$ will ever take over the package from another agent $j$ 
which $i$ overtakes from the left. In particular, this means that we will need at most one agent with starting position $p_i < s$,
and that after the package is picked up at $s$, it will never have to wait between a drop-off by an agent $j$ and a pick-up by the 
next agent $i$, since $j$ could continue carrying the package towards $i$, thus decreasing the overall delivery time. 
Hence in an optimum schedule we also have for consecutive agents $i,j$ with $s<p_j<p_i$, that $\pv_j\leq \pv_i$ 
(otherwise we can discard $i$, by this decreasing the delivery time).

\paragraph{Decomposition.} 
Assume that agent $i$ is the agent that delivers the package to $t$. 
We represent the trajectory of the package while being carried by $i$ as a ray giving the position $y$ on the real line as a function $f_i(x)$
of the time $x$ passed so far, see Figure~\ref{fig:time-decomposition} (right). 
We now inductively compute a set containing all functions $f_0, f_1, \ldots, f_k$, where $f_0(x) = s = 0$.

If we have $p_i < 0$, then by the reasoning above, $i$ is the only involved agent, and the function is simply $f_i\colon y = \pv_i\cdot x + p_i$. 
For $p_i > s$, the slope $\pv_i$ of the ray is set, but not its pick-up position. In order to minimize the earliest possible delivery time $x$ 
(i.e.~$f_i(x) = t$), by the non-decreasing velocity property $i$ must pick up the package as early as possible
-- e.g.~in Figure~\ref{fig:time-decomposition} (left), the fastest agent $6$ would not get the package from agent $4$, 
but from agent $5$ who is able to speed up the transport between agents $4$ and $6$, thus advancing the last handover position 
and allowing agent $6$ to pick up the package earlier.

Formally, the pick-up position is given by the time-wise first (or in other words leftmost) intersection of a query line $y=p_i-\pv_i\cdot x$ 
(modelling the agent moving towards $s$) with any preceding ray $f_0, \ldots, f_{i-1}$. 
Let $q_i := (x_i,y_i)$ denote the intersection point of the query line with the upper envelope of the preceding rays,
and denote by $f_j$ a ray of steepest slope $\pv_j$ among all rays $f_0, \ldots, f_{i-1}$ that contain $q_i$
-- e.g.~the query line ``$6?$'' in Figure~\ref{fig:time-decomposition} (right) intersects both $f_1$ and $f_5$ on the upper envelope,
and since both have the same slope, we can consider either.

In case $\pv_j>\pv_i$, agent $i$ will not be used in an optimal schedule and we set $f_i = 0$. If, however, $\pv_j \leq \pv_i$, 
then $f_i$ is given by the line equality $f_i\colon y = \pv_i\cdot x + (y_i-\pv_i\cdot x_i)$.
After completion, an optimum schedule corresponds to a path along the rays of our diagram from $(0,0)$ 
to the ray reaching $y=t$ at the earliest possible time.

\begin{figure}[t!]
	\includegraphics[width=\linewidth]{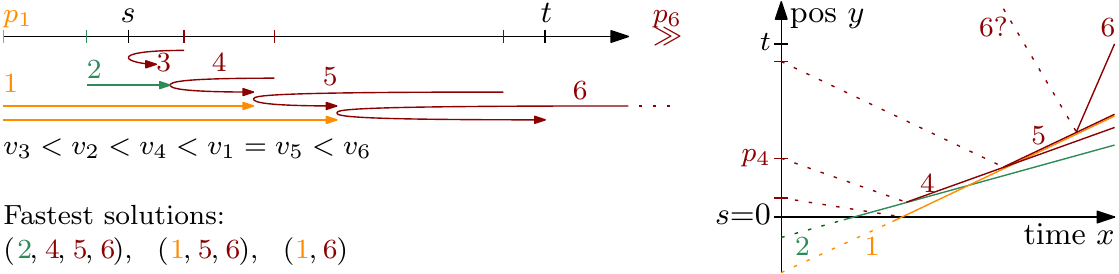}
	\caption[Time first: Line -- decomposition]{(left) Possible optima: 
	If agent $4$ is involved, it must take over the package from agent $2$, since agent $3$ is too slow. 
	Using agents $2$ and $4$ to bring the package to agent 5's pick-up position takes the same time as using agent $1$ on its own. 
	Agents $1$ and $5$ have the same velocity, so in terms of delivery time we could use either or even both of them, 
	but agent $1$ only if agents $2$ and $4$ are both not used (otherwise they all consume energy).
	(right) Fastest solutions correspond to at most $1$ agent with $p_i < s$ and a number of agents
	corresponding to a suffix of the upper envelope.
	}
	\label{fig:time-decomposition}
\end{figure}

\paragraph{Fast computation and recombination}
To quickly compute the equation of each ray $f_i$, we need to find the intersection of a query line with the 
upper envelope of $\bigO(k)$ many rays.
Precomputing this envelope as an ordered list of its segments would allow us to speed up the intersection queries
from a linear to a binary search (\emph{convex hull trick for dynamic programming}~\cite{wcipeg}).
However, the set of functions that we query here is not known up front. 
Instead, we apply the classic geometric point-line duality~\cite{edelsbrunner1987algorithms}.
In this dual setting, the task of finding the leftmost intersection point of a query line with a set of lines turns into 
finding a right tangent from a query point (the dual of the query line) onto the convex hull of a point set (the dual of the rays $f_i$). 
The dynamic planar convex hull data structure by Brodal and Jacob~\cite{RikoConvexHull02,jacob2002dynamic} 
allows point insertions and tangent queries all in $\bigO(\log k)$ amortized time, giving an overall running time of $\bigO(k\log k)$.
Assuming that we know the optimum schedule for each of the uniform velocity intervals,
it remains to recombine these subschedules:

\begin{lemma}
	Arbitrary velocity instances of \TWDelivery on paths can be decomposed into and recombined from 
	uniform velocity instances in time $\bigO(n+k\log k)$.
	\label{lem:decomposition}
\end{lemma}

\paragraph*{A fast algorithm for uniform velocity instances on the line}
We are left to solve the case where all agents have the same uniform velocity $\overline{\pv}$.
As before, we denote by $\delta$ the offset of the closest agent's starting position to $s$,
and let $a$ denote the corresponding agent. 
No agent $i$ with $p_i<s$ other than maybe agent $a$ is involved in an optimum schedule (all others would only slow down delivery). 
Also note that if $p_a < s$, the setting is equivalent to one where $a$ starts at $s + (s-p_a)$, 
so we can assume (after relabelling the agents) $a=1$, $\delta =p_1$.
This also implies $\T = (\delta + (t-s))/ \overline{\pv}$ and we can ignore agents $i$ that are dominated by earlier, 
cheaper agents $j$ with $p_j < p_i$ and $\pw_j < \pw_i$.

\paragraph{Towards a dynamic program.}
We define the point $q_i$ as the leftmost point on the line where agent $i$ can pick up the package without causing a delay, 
i.e.,~we have $q_i := \frac{p_i + s - \delta}{2}$ since $p_i - q_i = \delta + (q_i - s)$.
Note that $q_1 = s$ and $q_j < q_i$ for $j<i$.
Similarly as -- but more specific than -- in the characterization of uniform instances on general graphs 
(Lemma~\ref{lem:uniform-characterization}) we get a limited set of possible pick-up locations:

\begin{lemma}
	There is an optimum solution where each agent $i$ that is involved in advancing the package
	picks it up at $q_i^+ = q_i$ or at $q_i^+ = p_i$.	
	\label{lem:line-structure}
\end{lemma}

\paragraph{Dynamic program.} Lemma~\ref{lem:line-structure} suggests that in an inductive approach from left to right 
it suffices to consider only finitely many handover options. We define the following subproblems:
\begin{align*}
	S\left[ i \right] =	& \text{ An energy-optimal schedule to deliver the package to $p_i$}	\\
				& \text{ in time $(\delta + (p_i - s))/ \overline{\pv}$, 
					using only the agents $\left\{ 1,2,\ldots,i \right\}$.}\\
	\E\left[ i \right] =	& \text{ Energy consumption of $S\left[ i \right]$.}\\
	A\left[ i \right] =	& \text{ Index of the last package-carrying agent in $S\left[ i \right]$.}\\
	A'\left[ i \right] =	& \text{ Index of the second to last carrying agent in $S\left[ i \right]$ (if any).}
\end{align*}
We will argue how to compute the optimum energy costs $\E[i]$ (and with it $A[i]$ and $A'[i]$) 
without explicitly maintaining the schedules $S[i]$ ($S[i]$ can later be retraced from $A[i]$ and $A'[i]$).
For computing $\E[i]$, $A[i]$ and $A'[i]$, we distinguish four cases (also shown in Figure~\ref{fig:dp-cases}):
\begin{enumerate}
	\item	Agent $i$ is not involved in $S[i]$.
	\item	Agent $i$ is involved in $S[i]$.
		Hence by Lemma~\ref{lem:line-structure}, agent $i$ has pick-up location $q_i^+=q_i$; 
		and we get the following three variations:
		\begin{enumerate}
			\item	$i=1$ and agent $1$ picks up the package at $s$ itself.
			\item	Agent $i$ picks up the package from some other agent $j$ with $p_j \leq q_i$.
			\item	Agent $i$ picks up the package from some other agent $j$ with $p_j > q_i$.
		\end{enumerate}
\end{enumerate}

\begin{figure}
	\centering\includegraphics[width=\linewidth]{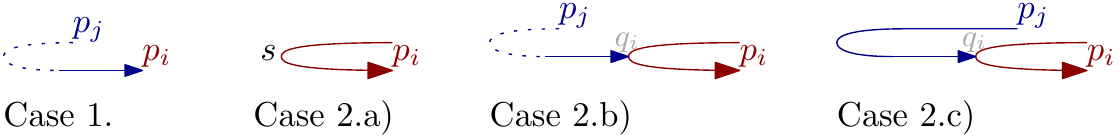}
	\caption{Case distinction in the dynamic program for \TWDelivery on the line. Either agent $i$ is not involved at all, does all on its own or is subsequent to some agent $j$, where we distinguish between $p_j \leq q_i$ and $p_j > q_i$.\label{fig:dp-cases}}
\end{figure}

\noindent
In cases 1, 2b) and 2c), we can determine $\E[i]$ in constant time using a single prior entry of the dynamic programming table:

\begin{description}
	\item[Case 1.]	If $i$ is not involved in $S[i]$, the best choice for the agent
		who transports the package to $p_i$ is agent $i-1$, 
		as it is the cheapest one on the last segment $[p_{i-1}, p_i]$ and we have 
		$\E[i-1] \leq \E[j] + (p_{i-1} - p_j) \cdot \pw_j$ for all $j<i-1$ by induction. 
		Hence we can optimize in constant time:
		\begin{align*}
			\E[i] &= \min_{j<i} \{ \E[j] + (p_i-p_j)\cdot \pw_j \} = \E[i-1] + (p_i-p_{i-1})\cdot \pw_{i-1},\\
			A[i] &= i-1 \text{ and } A'[i] = A[i-1].
		\end{align*}
	\item[Case 2.a)] This is the base case where the first agent is on its own:
		\begin{align*}
			\E[1] = 2\cdot \lvert p_1 - s \rvert \cdot \pw_1 = 2\cdot \delta \cdot \pw_1, 
			\quad A[1] = \text{none},\  A'[1] = \text{none}.
		\end{align*}
	\item[Case 2.b)] If agent $i$ is involved in $S[i]$ and takes over at $q_i$ from an agent $j$ with $p_j \leq q_i$, 
		we want $j$ to minimize $\E[j] + (q_i - p_j) \cdot \pw_j$, the cost of bringing the package to $q_i$. 
		Now let $i' = \max \{j \mid p_j \leq q_i\}$ be the agent starting closest to the left of $q_i$. 
		As in Case~1, we argue that $i'$ is the optimum choice for $j$, as it minimizes the cost on $[p_{i'}, q_i]$ 
		and does not constrain the schedule up to $p_{i'}$ further. 
		Hence, we again get in amortized constant time, i.e., we update $i'$ by incrementing it lazily when going from $i$ to $i+1$:
		\begin{align*}
			\E[i] &= \E[i'] + (q_{i}-p_{i'})\cdot \pw_{i'} + 2\cdot(p_i - q_i)\cdot \pw_i, \quad
			A[i] = i,\ A'[i] = i'.
		\end{align*}
\end{description}
The most interesting case is the remaining case $(2.c)$, where the agent $j$ handing over to $i$ starts in between $q_i$ and $p_i$.
Where can we look up the energy consumption $c$ of an optimum schedule that ends with $j$ bringing the package $q_i$ --
the dynamic program being only defined for points $p_j$?
For some $j$, we might have $A[j] = j$, so $S[j]$ ends by $j$ walking to $q_j$ and back. 
In that case, we can exploit $q_j < q_i$ and use $\E[j]-(p_j-q_i)\cdot \pw_j$ as a candidate for the energy consumption $c$.
But what if $A[j] \neq j$? As we saw, this implies $A[j] = j-1$,
but in that case we cannot just subtract $(p_j-q_i) \pw_j$: We do not know how $S[j]$ looks like between $q_i$ and $p_j$.
We argue that we do not need to consider these agents $j$ as candidates at all!

\begin{lemma}
	If in some optimal schedule $S[i]$ the agent $j$ preceding $i$ is of type 2.c), 
	then in the schedule $S[j]$ we have $A[j] = j$.
	\label{lem:line-interesting-agents}	
\end{lemma}
\begin{proof}
	Under the assumption of Case~2.c), the cost of agent $i$ is fixed to $2 \cdot (p_i - q_i) \cdot \pw_i$.
	Agents $1$ to $i-1$ will collaborate in the most efficient way to bring the package up to $q_i$.
	By definition of $j$, $j$ is the last agent bringing the package to $q_i$.
	From the decreasing weight property, we know that none of the agents $j+1$ to $i-1$ were involved in $S[i]$.
	So if we take the partial schedule of $S[i]$ up to $q_i$ and extend it by letting $j$ bring the package to $p_j$, 
	we obtain a feasible candidate schedule $S'$ for $S[j]$ as none of the agents $j+1$ to $i$ are involved.
	We now argue that $S'$ is an optimum schedule for $S[j]$.
	The segment $[q_i,p_j]$ is covered with the minimum possible energy, 
	as $j$ is the \emph{unique} most efficient agent available for $S[j]$.
	The segment $[s, q_i]$ is also covered cheapest possible as its part of $S[i]$ was optimized over all agents $1$ to $i$, 
	so a superset of the agents available for $S[j]$.
	Moreover, the uniqueness implies that \emph{all} optimum schedules for $S[j]$ need to end with agent $j$ on $[q_i,p_j]$, 
	hence $A[j] = j$.
\end{proof}

\begin{description}
	\item[Case 2.c)] Lemma~\ref{lem:line-interesting-agents} leaves us with only those agents $j$ 
		whose schedules $S[j]$ we understand sufficiently to modify them into candidates for $S[i]$ under Case~2.c):
		\begin{align*}
			\E[i] & = \min_{j} \{\E[j] - (p_j - q_i) \pw_j \mid q_i < p_j < p_i \wedge A[j]=j \} 
				 + 2(p_i - q_i)\pw_i,\\
			A[i]  & = i,\ \,
			A'[i] = \argmin_{j} \{\E[j] - (p_j - q_i)\pw_j \mid q_i < p_j < p_i \wedge A[j]=j \}.
		\end{align*}
\end{description}
We can now take $\E[i]$ as the minimum over the four cases 1--2.c) and compute all schedules $S[i]$ by proceeding over 
all subproblems in increasing order, giving us the energy-optimal schedules for delivering the package to the points $p_i$
in time $(\delta+(p_i-s))/\overline{\pv}$.
How can we use the solutions to the subproblems $\E[i]$ to find the energy $\E$ of an energy-optimal schedule 
delivering the package to the target $t$ in optimum time $(\delta+(t-s))/\overline{\pv}$? 
Let $k'$ be the closest agent on the left of $t$, i.e., $k' := \argmax_i {p_i\leq t}$.
Clearly, if in an optimum schedule the package is delivered to $t$ by an agent starting to the left of $t$, then by the decreasing 
weight property this agent must be agent $k'$, giving us $\E = \E[k'] + (t-p_{k'})\cdot \pw_{k}$.

\paragraph{Delivery to $t$ and agents with $p_i > t$.}
It remains to take care of agents with starting positions $p_i>t$:
\begin{figure}
	\centering\includegraphics[width=\linewidth]{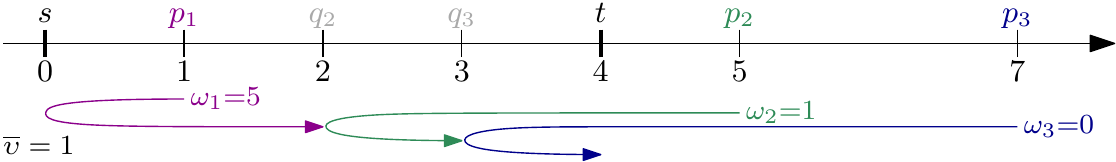}
	\caption{An example, where the only optimum schedule uses both agents on the right of $t$.
		The optimum schedule has delivery time $\T = (p_1 + t - s)/\overline{\pv} = 5/1=5$ and energy consumption 
		$\E = p_1 \cdot \pw_1 + (p_2-q_2 + q_3-q_2)\cdot \pw_2 + (p_3-q_3 + t-q_3)\cdot \pw_3
	= 3 \cdot 5 + 4 \cdot 1 + 5 \cdot 0 = 19$.\label{fig:time:right-example}}
\end{figure}
As illustrated in Figure~\ref{fig:time:right-example}, multiple agents with $p_i > t$ might be involved in the most efficient delivery.
Note that our dynamic programming problem $\E[i]$ is defined independent of $t$ and so we can also easily compute $\E[i]$ for $p_i > t$.
Agents $i$ with $q_i > t$ are not useful, however, for a delivery to $t$, as they arrive in $[s,t]$ only after the package has been delivered.
Similar to Lemma~\ref{lem:line-interesting-agents}, 
we claim that among the remaining agents $i$ only those with $A[i] = i$ need to be considered:

\begin{lemma}
	If an agent $i$ with $p_i > t$ is the last agent in any optimal schedule $S$ from $s$ to $t$, then $A[i]=i$.
	\label{lem:line-right-agents}	
\end{lemma}
\begin{proof}
	We have $q_i < t < p_i$. By the decreasing weight property, no agent $j>i$ will be used in $S$.
	We extend $S$ to a schedule $S'$ by letting agent $i$ walk from $t$ to $p_i$. Then $S'$ is a candidate for $S[i]$.
	Similar to Lemma~\ref{lem:line-interesting-agents}, $S'$ consists of an optimal solution for $[s,t]$ 
	and the strictly cheapest agent on $[t,p_i]$ and hence $S'$ is optimal for $S[i]$ and all optimum schedules for $S[i]$ have $A[i]=i$.
	The optimum $s$-$t$-delivery is thus given by:
	\begin{align*}
		\E = \min_{j} \left\{ \E[j] - (p_j-t) \pw_j \mid \left( q_j < t < p_j \text{ and } A[j]=j \right) \text{ or } j = k' \right\},
	\end{align*}
	which takes linear time once at the very end.
\end{proof}

\paragraph{Details of the dynamic program.}
The computational bottleneck of our dynamic program is (for each subproblem $\E[i]$)
the minimization over the set of options in Case~2.c). Each option evaluates a linear function 
$f_j(q_i):= \pw_j \cdot q_i + (\E[j] - p_j \cdot \pw_j)$ at position $q_i$, which can be seen as a lower envelope intersection query.
Similarly to before, we use point-line duality and a dynamic convex hull data
structure to avoid considering all agents explicitly as predecessors and instead quickly search the best one.

\begin{lemma}
	An optimum schedule for \TWDelivery with uniform velocity $\overline{\pv}$ on the line can be computed in $\bigO(n + k \log k)$ time.
	\label{lem:dynamic-dp}	
\end{lemma}

Combining this with Lemma~\ref{lem:decomposition} gives the full solution on paths. 
Note that strictly speaking, in the uniform velocity instances, the package is not available at $s$ at time zero, 
but is brought there by agents of preceding instances at exactly the time when the first agent can reach it.

\begin{theorem}	
	An optimum solution for \TWDelivery on paths can be computed in $\bigO(n + k \log k)$ time.
	\label{thm:path}
\end{theorem}

\section{Optimizing convex combinations of objectives}
\label{sec:combination}

In this section, we look at a convex combination of the two objectives: 
minimizing both the delivery time $\T$ and the energy consumption $\E$ by minimizing the term $\epsilon\cdot \T+(1-\epsilon)\cdot \E$, 
for a given value $\epsilon,\ 0 < \epsilon < 1$. We call the problem of minimizing this combined objective \CDelivery.
As an application of the \NP-hardness proof for \TWDelivery, we get \NP-hardness of \CDelivery as well:
The main idea is to counter small values of $\epsilon$ by scaling the weights of the agents by a small factor $\delta(\epsilon)$, 
thus decreasing the importance of $\E$ alongside $\T$ as well.
\begin{theorem}
	\CDelivery is \NP-hard for all $\epsilon \in (0,1)$, even on planar graphs.	
	\label{thm:combination-hardness}
\end{theorem}

\paragraph*{A $3$-approximation for \CDelivery using a single agent}
Recall that for \TDelivery, the agents involved in an optimum delivery were characterized by increasing velocities $\pv_i$, 
while for \TWDelivery on path graphs, the agents of an optimum solution were characterized by decreasing tuples $(\pv_i^{-1},\pw_i)$.

Although it is \emph{not} possible to characterize the \emph{order} of the agents in an optimum \CDelivery schedule 
by their velocities and weights alone, 
we can at least characterize the position of a \emph{minimal} agent, leading to a 3-approximation using a single agent:
\begin{lemma}
	Let without loss of generality $1, 2, \ldots, i$ denote the indices of all involved agents appearing in that order 
	in an optimum \CDelivery schedule. 
	Then the last agent $i$ is \emph{minimal} in the following sense:
	$ i \in \arg \min_{j} \smash{\left\{ \epsilon\cdot \pv_{j}^{-1} + (1-\epsilon)\cdot \pw_{j} \right\}}$.
	\label{lem:combined-minimal}
\end{lemma}
\begin{proof}
	Recall that we denote by $d_{j}$ the total distance travelled by agent $j$ and by 
	$d_j^*$ the distance travelled by agent $j$ while carrying the package. 
	Thus agent $j$ contributes at least 
	$\epsilon\cdot d_j^*\cdot \pv_j^{-1} + (1-\epsilon)\cdot \left(d_j\cdot \pw_j\right) \geq 
	d_j^*\cdot\left( \epsilon \pv_{j}^{-1} + (1-\epsilon)\pw_{j}  \right)$ towards $\epsilon\T+(1-\epsilon) \E$.
	Assume for the sake of contradiction that the minimum value $\epsilon \pv_{j}^{-1} + (1-\epsilon) \pw_{j}$ is not obtained by agent $i$ 
	but by an agent $m<i$. Then we can replace the agents $m+1, \ldots, i$ by agent $m$, resulting in a decrease in the objective function
	of at least
 	\begin{align*}
 		& \sum\limits_{j=m+1}^i \epsilon d_j^* \pv_j^{-1} + (1-\epsilon)d_j \pw_j 
 			- \sum\limits_{j=m+1}^i d_j^*\left( \epsilon \pv_{m}^{-1} + (1-\epsilon)\pw_{m}  \right) \\
 		& \geq \sum\limits_{j=m+1}^i d_j^*\left( \epsilon \pv_{j}^{-1} + (1-\epsilon)\pw_{j}  \right) 
 			- \sum\limits_{j=m+1}^i d_j^*\left( \epsilon \pv_{m}^{-1} + (1-\epsilon)\pw_{m}  \right) \\
 		& \geq d_i^* \left( \epsilon \pv_i^{-1} + (1-\epsilon)\pw_i \right) - d_i^* \left( \epsilon \pv_m^{-1} + (1-\epsilon)\pw_m \right) > 0,
 	\end{align*}
	contradicting the minimality of the optimum \CDelivery schedule.
\end{proof}

\begin{theorem}
	There is a $3$-approximation for \CDelivery which uses only a single agent (and thus can be found in polynomial time).	
	\label{thm:3-approx}
\end{theorem}
\begin{proof}
	Note that agent $i$ contributes at most $\epsilon d_i\pv_i^{-1} + (1-\epsilon) d_i\pw_i$ towards $\epsilon\T+(1-\epsilon) \E$. 
	Starting from an optimum \CDelivery schedule we can replace all agents $1, \ldots, i-1$ along their trajectories by the minimal agent $i$. 
	This prolongs the travel distance of agent $i$ by $2\cdot \sum_{j=1}^{i-1} d_j^*$. Overall, we increase the objective function by at most
	\begin{align*}
		2\sum\limits_{j=1}^{i-1} d_j^* \left(\epsilon\pv_i^{-1} + (1-\epsilon) \pw_i \right)
		\leq 2 \sum\limits_{j=1}^{i} d_j^*\left( \epsilon \pv_{j}^{-1} + (1-\epsilon)\pw_{j}  \right) \leq 2 \left( \epsilon\T+(1-\epsilon) \E\right). 
	\end{align*}
	Hence only using agent $i$ to deliver the package is a $3$-approximation for \CDelivery.
	We get a polynomial-time approximation algorithm with approximation ratio $3$ by choosing among all $k$ agents the one with minimum value 
	$\left( \epsilon \pv_j^{-1} + (1-\epsilon)\pw_{j} \right)\cdot \left( d(p_j,s) + d(s,t) \right)$.
\end{proof}

\section{Discussion}
\label{sec:discussion}

Our techniques and results extend to a variety of delivery problems and model generalizations.
A key ingredient here is that the order of our dynamic programming subproblems depends only on the 
parameters the agent has \emph{while carrying the package}. Hence it is possible to, e.g., incorporate
2-speed agent models (modeling different speeds~\cite{Czyzowicz16} with/without carrying the package) 
or topographical features (modeling edge traversals in uphill/downhill direction).

Furthermore, the 3-approximation given for \CDelivery is applicable to \TDelivery as well, 
and a relaxation of \TWDelivery, in which one allows the optimum delivery time to be achieved with a 
constant-factor approximation of the energy consumption, stays \NP-hard. 
It is unclear whether \TWDelivery can be \emph{solved efficiently on trees} or whether \CDelivery \emph{allows a PTAS}.
We consider these two problems the major open questions raised by this work.

\bibliographystyle{abbrv}
\bibliography{bib-fastefficient}

\appendix
\newpage
\section{Existence of optimum solutions}

\begin{theorem*}[\ref{thm:existence} \textnormal{(Existence of optimum solutions)}]
There exists an optimum solution minimizing the delivery time $\T$ 
(the energy consumption $\E$, or $\epsilon\cdot \T +(1-\epsilon)\cdot \E$, $(\T,\E)$, $(\E,\T)$, respectively).
\end{theorem*}

\begin{proof}
	A solution which operates agents $1,2,\ldots,\ell$ (after relabeling) in this order is uniquely described by 
	the order of those agents and their drop-off locations $q_i^-$:
	Agent $1$ picks up the package at $s$ and drops off the package at $q_{1}^-$, 
	agent $2$ picks it up at $q_{1}^-$ and drops it off at $q_{2}^-$, \ldots, 
	and agent ${\ell}$ drops off the package at $q_{{\ell}}^- = t$.

	To each drop-off location we assign an energy consumption- and a time-measure:
	Let $\E(q_i^-)$ and $\T(q_{i}^-)$ denote the total energy consumption (time, respectively) spent
	up to the point where agent $i$ drops the package at $q_{i}^-$.
	These values are inductively defined as follows:
	\begin{align}
		\E(q_{1}^-) &:= \pw_{1} \cdot (d(p_1,s) + d(s,q_1^-)), \qquad \E(S) = \E(q_{\ell}^-), \nonumber \\
		\E(q_{i}^-) &:= \E(q_{i-1}^-) + \pw_{i} \cdot (d(p_i,q_{i-1}^-) + d(q_{i-1}^-,q_i^-)), \label{eq:E}
	\end{align}
	and
	\begin{align}
		\T(q_{1}^-) &:= \pv_{1}^{-1} \cdot (d(p_1,s) + d(s,q_1^-)), \qquad \T(S) = \T(q_{\ell}^-), \nonumber \\
		\T(q_{i}^-) &:= \max\left\{ \T(q_{i-1}^-), \pv_i^{-1}\cdot d(p_i,q_{i-1}^-) \right\} + \pv_i^{-1}\cdot d(q_{i-1}^-,q_i^-)).
		\label{eq:T}
	\end{align}
	Because there are infinitely many possible (in-edge-)handover positions, there are also infinitely many solutions.
	We will now show that we can subdivide these solutions into a finite number of equivalence classes. We prove that 
	each equivalence class has
	\begin{itemize}
		\item	a solution that minimizes the delivery time $\T$ among all solutions in the class, 
		\item	a solution that minimizes the energy consumption $\E$ among all solutions in the class.
	\end{itemize}
	Since the number of equivalence classes itself is finite, we immediately get that the (infinite) set of all solutions
	itself contains minimum solutions subject to $\T$ and $\E$, respectively.
	We may think of the subdivision into equivalence classes as a partition, although strictly speaking some schedules 
	will be contained in multiple distinct equivalence classes.

	\paragraph{Representation via drop-off positions}
	We first group all solutions into schedules that operate the same agents in the same order. There is a finite number
	of such groups. Now recall that for a specific list of agents $1,2,\ldots,\ell$ a solution is uniquely described by the
	$\ell$ many drop-off positions $q_1^-, \ldots, q_{\ell-1}^-$ (since we must have $q_{\ell}^-=t$).  

	We represent a drop-off $q_i^-$, $i < \ell$ in an edge $\{u,v\}$ as a tuple $((u,v),x_i)$, where $x_i \in \left[0,1\right]$ 
	is a parameter denoting the position of $q_{i}^-$ in $\{u,v\}$, $x_i := \tfrac{d(u,q_i^-)}{d(u,v)}$. 
	If $q_i^-$ lies strictly inside the edge $\{u,v\}$, this representation is unique. 
	The same is true if $q_i^-$ is a node of degree 1. 
	If $q_i^-$ is a node of degree $\deg(q_i^-)>1$, there are $\deg(q_i^-)$ many representations to choose from. 

	Similarly to individual handover points, we can represent schedules in this parametrized notation, too. 
	For each schedule $S$ operating the agents $1,\ldots,\ell$ in this order, we consider all possible representations 
	of the whole schedule:
	If $S$ has $x$ many node drop-offs $v_1,\ldots,v_x$, then $S$ has a bounded number of exactly 
	$\prod_{i=1}^x \deg(v_i)$ many different representations.     

	\paragraph{Equivalence relation between representations}
	We can now define an equivalence relation between parametrized representations of schedules:
	Two parametrized representations $S_P$ and $S'_P$ of a solution $S$ and a solution $S'$, respectively, are 
	equivalent, $S_P \prel S'_P$, if they operate the same set of agents in the same order and agree 
	in the edge of each drop-off location $q_{i}^-$ (in other words, $S_P$ and $S'_P$ differ only in 
	the exact positions inside the edges, given by the respective parameters $x_i$). 

	Note that two different parametrized representations of the same schedule are not equivalent. 
	Since a parametrized representation $S_P$ of a schedule $S$ preserves all information contained in $S$, 
	we can also measure $S_P$ subject to $\T$ and $\E$. 

	\paragraph{Minima of equivalence classes}
	We are now ready to prove that, given an arbitrary parametrized schedule representation $S_P$, 
	the equivalence class of $S_P$ under $\prel$, denoted by $[S_P]$, contains a minimum element subject to $\T$
	and a minimum element subject to $\E$.

	A schedule $S$ operating $\ell$ agents has $\ell-1$ parametrized drop-off locations. 
	If $\ell-1 = 0$, then $[S_P]$ contains the single element $S_P$ and thus has a minimum element. 
	Otherwise, any different choice of the values $x_i\in \left[0,1\right]$ of these parameters 
	still represents a solution, and all these solutions are in the equivalence class $[S_P]$. 
	Hence we can identify $\left[ S_P\right]$ with the 
	$\ell-1$-dimensional hypercube $\left[0,1\right]^{\ell-1}$ which is a bounded and closed metric space
	and thus a topologically \emph{compact} space.

	Since we have $\E(S) = \E(q_{\ell}^-)$ and $\E(T) = \T(q_{\ell}^-)$, using Equations~\eqref{eq:E} and~\eqref{eq:T} 
	together with the triangle inequality for distances in a graph, it is easy to see that $\E$ and $\T$ are continuous 
	functions mapping parametrized schedule representations in the compact hypercube $\left[0,1\right]^{\ell-1}$ to the 
	corresponding energy-consumption and delivery-time values in $\mathbb{R}$:

	Consider the minimum velocity $\pv_{\min} = \min_i \pv_i$, the maximum weight $\pw_{\max} = \max_i \pw_i$ and the maximum 
	edge length $l_{\max} = \max_e l_e$. We use the $L1$-norm and show that for all $\epsilon > 0$ there exists a 
	$\delta = \delta(\epsilon) := \smash{\epsilon \cdot \min\{ \pv_{\min}, \tfrac{1}{\pw_{\max}l_{\max}}\} \cdot (2\ell)^{-1}}$, 
	such that we have for any two parametrized schedule representations $S'_P, S''_P \in \left[0,1\right]^{\ell-1}$ that
	\[	\left\lVert S'_P - S''_P \right\rVert_1 < \delta 
		\ \Rightarrow\ \left\lvert \E(S'_P)-\E(S''_P) \right\rvert < \tfrac{\epsilon}{2} \quad \wedge \quad	
		\left\lvert \T(S'_P)-\T(S''_P) \right\rvert < \tfrac{\epsilon}{2}. \]
	Note that from $\left\lVert S'_P - S''_P \right\rVert_1 < \delta$ we get for any two parametrized drop-off $q'^-_i, q''^-_i$ that
	$|x'_i - x''_i| < \delta$ and thus 
	$|q'^-_i - q''^-_i| < l_{\max}\cdot |x'_i - x''_i| = \epsilon \cdot \min\{ \pv_{\min}, \tfrac{1}{\pw_{\max}}\} \cdot (2\ell)^{-1}$.
	Using Equations~\eqref{eq:E} and~\eqref{eq:T} inductively together with the triangle inequality for graph distances we get
	\begin{align*}
		|\E(q'^-_i) - \E(q''^-_i)|	& < |\E(q'^-_{i-1}) - \E(q''^-_{i-1})| + 2\cdot \pw_{\max}\cdot |q'^-_i - q''^-_i| 	
						\leq \frac{\epsilon \cdot i}{\ell} && \text{ and }\\
		|\T(q'^-_i) - \T(q''^-_i)|	& < |\T(q'^-_{i-1}) - \T(q''^-_{i-1})| + 2\cdot \tfrac{1}{\pv_{\min}}\cdot |q'^-_i - q''^-_i| 	
						\leq \frac{\epsilon \cdot i}{\ell}.
	\end{align*}
	We have $\left\lVert S'_P - S''_P \right\rVert_1 < \delta  \Rightarrow |\E(S'_P) - \E(S''_P)| < \epsilon$ and  
	$\left\lVert S'_P - S''_P \right\rVert_1 < \delta \Rightarrow |\T(S'_P) - \T(S''_P)| < \epsilon$, and thus $\E$ and $\T$ are 
	continuous functions. Hence by the extreme value theorem, there must be a solution of minimum value in $\left[0,1\right]^{\ell-1}$ with 
	respect to $\E$ and to $\T$, respectively. 
	The same holds for the linear combination $\epsilon \T +(1-\epsilon)\E$, and, furthermore, there are also 
	lexicographically minimum solutions with respect to the tuples $(\T,\E)$ and $(\E,\T)$, 
	respectively~\cite{matousek}.

	\vspace{2ex}
	\noindent
	Finally, since there is only a finite number of equivalence classes, there must be minimum solutions for all five
	variants of efficient \Delivery overall.			
\end{proof}

\newpage
\section{\NP-hardness of \TWDelivery}

\paragraph{Planar 3SAT.} 
We start with a three-conjunctive normal form $F$ on $x$ variables $V(F) = \left\{ u_1, \ldots, u_x \right\}$ 
and $y$ clauses $C(F) =  \left\{ c_1, \ldots, c_y \right\}$. Each clause is a disjunction of at most three literals 
of the form $l(u_j) \in \left\{ u_j, \overline{u_j} \right\}$.
$F$ can be represented by a graph $H(F) = (C(F) \cup V(F),A_1 \cup A_2)$ defined as follows: 
$(C(F) \cup V(F),A_1)$ is a bipartite graph with nodes corresponding to all clauses and all variables and an edge set 
$A_1$ which contains an edge between each clause $c$ and variable $u$ if and only if $u$ or $\overline{u}$ is contained in $c$,
$A_1 = \left\{ \left\{ c_i, u_j \right\} \ | \ u_j \in c_i \text{ or } \overline{u_j} \in c_i \right\}$.
To this graph we add a cycle $A_2$ consisting of edges between all pairs of consecutive variables, 
$A_2 = \left\{ \left\{ u_j, u_{(j\mod x)+1} \right\} \ | \ 1\leq j \leq x \right\}.$
The 3CNF $F$ is called \emph{planar} if there is a plane embedding of $H(F)$ which \emph{at each variable node} has 
all edges representing positive literals on one side of the cycle $A_2$ and 
all edges representing negative literals on the other side of $A_2$.
The decision problem \textsc{Planar3SAT} of finding whether a given planar 3CNF $F$ is satisfiable or not is \NP-complete, 
a result due to Lichtenstein~\cite{planar3sat82}.
We assume without loss of generality that every clause contains at most one literal per variable, see Figure~\ref{fig:planar3sat} (left).

\begin{figure}[bp]
	\centering
	\includegraphics[width=\linewidth]{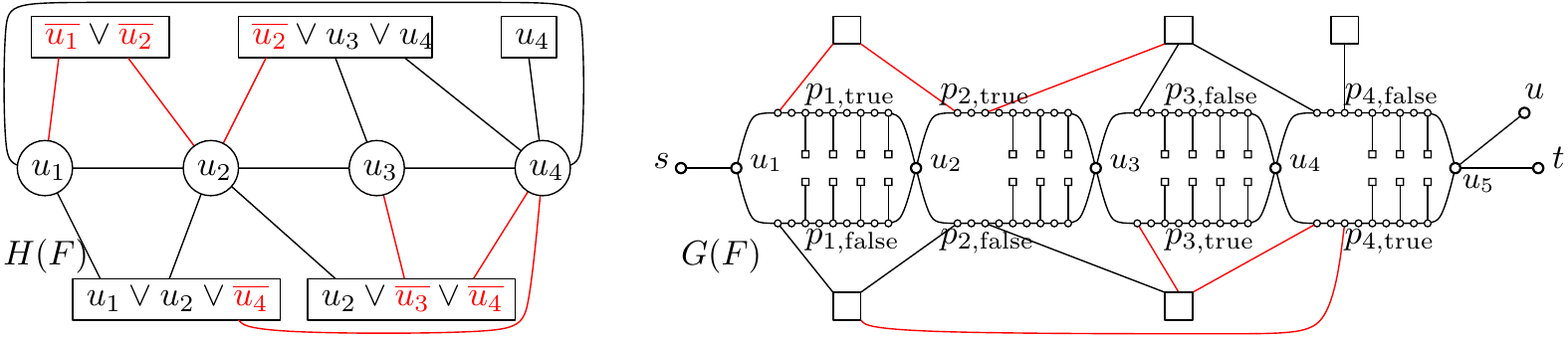}
	\caption{(left) An example of a planar 3CNF $F$ with a satisfiable assignment $(u_1,u_2,u_3,u_4) = (\text{true, false, false, true}).$ 
	\hfill (right) Its transformation into the corresponding delivery graph $G(F)$.}
	\label{fig:planar3sat}
\end{figure}

\paragraph{Delivery graph.} We now transform $H(F)$ into a planar graph $G(F)$ on which we will build an instance of \TWDelivery.
First, we place an additional node $u_{x+1}$ on the edge $\left\{ u_x, u_1 \right\}$ and modify $A_2$ accordingly 
(by adding $\left\{ u_x, u_{x+1} \right\}, \left\{ u_{x+1}, u_{1} \right\}$ to and removing $\left\{ u_{x}, u_1 \right\}$ from $A_2$).
In a second step, we add on both sides of each edge $\left\{ u_j, u_{j+1} \right\}, 1\leq j \leq x,$ 
a path from $u_j$ to $u_{j+1}$ with $2y-1$ internal nodes.
We can now reconnect the literal edges $\left\{ c_i, u_j \right\}$ while preserving planarity: 
If $\left\{ c_i, u_j \right\}$ is the $l$th literal edge of $u_j$ inside the cycle $A_2$, 
we reconnect $c_i$ to the $(2l-1)$th internal node of the path from $u_j$ to $u_{j+1}$ which lies inside $A_2$. 
We do the same for literal edges outside of $A_2$. 
For every \emph{odd} node on these paths which is not yet connected to a clause node, 
we create a single new node and connect it to the internal path node. 
Thus, for all $j$, both paths from $u_j$ to $u_{j+1}$ have $y$ internal nodes of degree $3$ 
and $y-1$ internal nodes of degree $2$, and these nodes alternate. 
If the path contains nodes connected to clauses which contain $u_j$, we call the path $p_{j,\text{false}}$, 
otherwise we call it $p_{j,\text{true}}$.
Finally, we delete all edges of $A_2$, add the package's source $s$ and destination $t$ and another node $u$ 
and connect these to $u_1$, $u_{x+1}$ and $u_{x+1}$, respectively. 
We redraw the graph such that $s, u_1, \ldots, u_{x+1}, t$ are placed in this order from left to right, 
see Figure~\ref{fig:planar3sat} (right).

\paragraph{Reduction idea.} 
We will place agents and set edge lengths such that any delivery of the package which goes via any of the clause nodes takes a long time. 
Thus the package has to be routed in each path pair $(p_{j,\text{true}}, p_{j,\text{false}})$ through exactly one of the two paths. 
If the package is routed via the path $p_{j,\text{true}}$, we interpret this as setting $u_j=\text{true}$ and hence we can read from the
package trajectory a satisfiable assignment for $F$. 

\paragraph{Agent placement.} We use three kinds of agents: 
slow but energy-efficient agents $(\pv=1,\pw=0)$, fast but inefficient agents $(\pv=2,\pw=1)$ and one very fast agent $(\pv=8, \pw=0)$:
\begin{itemize}
	\item	The fastest agent is placed on $u$ and shall transport the package over the edge $\left\{ u_{x+1}, t \right\}$.
	\item	Fast agents $(\pv=2,\pw=1)$ on one hand are placed on each variable node $u_j$. 
		These agents will decide whether to deliver the package over the path $p_{j,\text{true}}$ or the path $p_{j,\text{false}}$, 
		thus effectively setting the value of the corresponding boolean variable $u_j$ to true or to false. 
		On the other hand, we will also use one fast agent on each internal path node of degree $2$.
	\item	Finally, slow agents are placed as follows: 
		On each clause node $c_j$ of degree $\deg(c_j)$ we place $\deg(c_j)-1$ slow agents. 
		We think of these agents as follows: If a clause $c_j$ consists of three literals, 
		e.g. $(u_2 \vee \overline{u_3} \vee \overline{u_4})$,
		then in a satisfiable assignment at most two of these literals (e.g. $u_2$ and $\overline{u_4}$) are evaluated to false. 
		Thus the paths $p_{2,\text{false}}$ and $p_{4,\text{true}}$ are used in the delivery
		and the $\deg(c_j)-1=3-1=2$ agents are sent towards these paths to help carry the package. 
		Along the same lines, on each newly created node that is connected to an internal node of a path we place a single slow agent. 
		Additionally, a slow agent is placed on the package source $s$.
\end{itemize}
Overall we have (i) two agents on nodes $s$ and $u$, (ii) $x$ fast agents on the variable nodes, 
(iii) $2x(y-1)$ fast agents on the paths $p_{j,\text{true}}, p_{j,\text{false}}$ (one for each internal node of degree two),
and (iv) for each of the $2xy$ internal path nodes of degree 3 (except for $y$ many) a unique slow agent on an adjacent node. 
In total we get $k=4xy-x-y+2$ mobile agents.

\paragraph{Edge lengths.} Recall that we started from $x$ variables and $y$ clauses.
We set the length of the first edge $\left\{ s, u_1 \right\}$ to $12x^2y^2$, 
and the lengths of the edges $\left\{ u_{x+1}, u \right\}, \left\{ u_{x+1},t \right\}$ to $128x^2y^2$.
The edges along paths $p_{j,\text{true}}$ and $p_{j,\text{false}}$ 
have length $2$ if they lie directly to the right of a fast agent's starting position
and length $4xy-1$ if they are adjacent on its left.
It remains to set the length of all remaining edges adjacent to internal path nodes (connecting to either a clause node or a newly created node).
Each such edge, adjacent to the $(2l-1)$th internal node (counted from the left) of path $p_{j,\text{true}}$ or $p_{j,\text{false}}$, gets length
$12x^2y^2 + (j-1)\cdot 4xy^2 + (l-1) \cdot 4xy + 1$, see Figure~\ref{fig:lengths}.

From the defined edge lengths we get that each path $p_{j,\text{true}}$ and each path $p_{j,\text{false}}$ has length exactly $y\cdot(2+(4xy-1)) = 4xy^2+y$
and thus $d_G(u_1,u_{x+1}) = x\cdot d_G(u_j, u_{j+1}) = 4x^2y^2+xy$. 

\begin{lemma}[minimum delivery time]
	\label{lem:minimum-time}
	Any fastest delivery of the package from $s$ to $t$ takes time $\T = 32 x^2y^2$ and does not go via any clause node of $G(F)$.
\end{lemma}
\begin{proof}
	Clearly, any delivery which does not use the fastest agent (with $\pv=8$) takes time at least $128x^2y^2 / 2 = 64x^2y^2$. 
	Hence the fastest agent has to deliver the package to $t$ and thus travel from its starting position $u$ to $u_{x+1}$ 
	and later on from $u_{x+1}$ to $t$. We get $\T \geq 2\cdot 128x^2y^2 / 8 = 32x^2y^2$.

	Equality holds if and only if the other agents can collectively deliver the package to $u_{x+1}$ in time $16x^2y^2$. 
	This is the case if the slow agent at $s$ brings the package to $u_1$ in time $12x^2y^2/1 = 12x^2y^2$, 
	where the fast agent takes over and brings the package to $u_{x+1}$ in time $d_G(u_1,u_{x+1})/2 = (4x^2y^2+xy)/2 < 4x^2y^2$.

	If, however, the trajectory of the package contains a clause node $c_j$, then the package travels at least a distance of 
	$d_G(s,u_1) + d_G(u_1,c_j) + d_G(c_j, u_{x+1}) \geq 3\cdot 12x^2y^2$ which takes time at least $36x^2y^2/2 = 18x^2y^2$. 
	Note that in the last step we assumed that we do not use the very fast agent of velocity $8$ -- 
	if we used that agent, it would travel a total distance of strictly more than $2\cdot 128x^2y^2$ and hence $\T > 32x^2y^2$ as well.
\end{proof}

\noindent We now show that among all fastest delivery schedules (of time $\T = 32x^2y^2$) there exists a schedule 
with energy consumption $\E = 2xy$ if and only if $F$ is satisfiable (otherwise $\E > 2xy$).

\begin{figure}[tp]
	\centering
	\includegraphics[width=\linewidth]{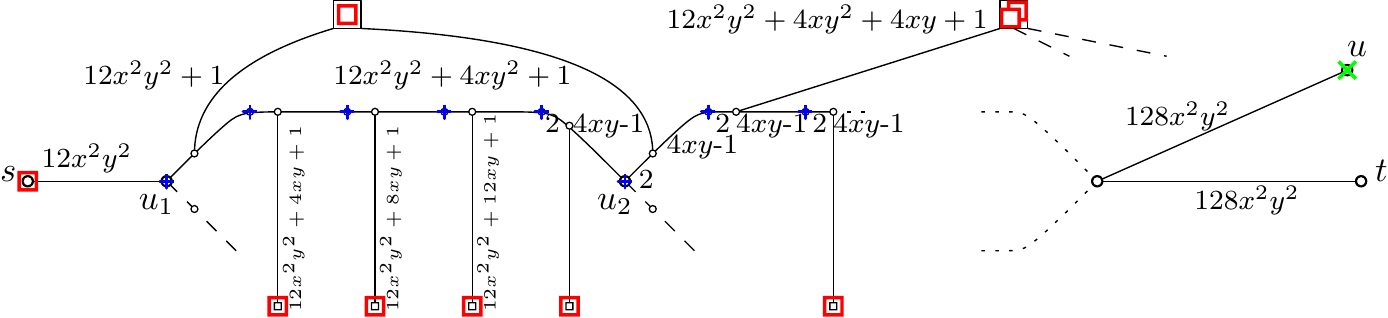}
	\caption{Edge lengths; placement of slow agents (\textcolor{red}{$\square$}), fast agents (\textcolor{blue}{$+$})
	and a very fast agent (\textcolor{green}{$\times$}).}
	\label{fig:lengths}
\end{figure}

\begin{lemma} [$(\T,\E)$ of a \psat-solution]
	Given a satisfiable assignment (a solution) for the variables of a 3CNF $F$
	there is a delivery schedule with delivery time $\T=32x^2y^2$ and energy consumption $\E=2xy$.	
	\label{lem:sat-delivery}
\end{lemma}
\begin{proof}
	We first remark that the slow agent on $s$ can bring the package to $u_1$ in time $12x^2y^2$ with energy consumption $0$
	and that the fastest agent on $u$ can get to $u_{x+1}$ in time $16x^2y^2$ with energy cost $0$ 
	(and from there reach $t$ in time $16x^2y^2$ and energy $0$). 
	It remains to show that all other agents can collectively deliver the package from $u_1$ to $u_{x+1}$ in time $4x^2y^2$ 
	such that the package never has to wait and without spending more energy than $2xy$.

	Given a satisfiable assignment for the variables in $F$, the actions of the fast agents are straightforward:
	Each agent placed on a variable node $u_j$ moves according to the variable assignment of $u_j$ along the first edge
	of either the \emph{true}-path $p_{j,\text{true}}$ or the \emph{false}-path $p_{j,\text{false}}$
	as soon as the package arrives at $u_j$. 
	The other fast agents on internal path nodes of degree $2$ also wait for the package and upon arrival carry it over the 
	adjacent edge of length $2$. Collectively, each of the $xy$ fast agents (with velocity $2$ and weight $1$) that move 
	take time $2/2=1$ and energy $2\cdot 1=2$ to cross the adjacent edge. 

	Finally, we show that the package can be carried over the remaining gaps (edges of length $4xy-1$) by one slow agent each, 
	stationed on an adjacent clause or newly created node, \emph{without} the package having to wait for the agent to arrive.
	If this is the case, then $y$ slow and $y$ fast agents carry the package over a path from $u_j$ to $u_{j+1}$ in time 
	$y\cdot (4xy-1)/1 + y\cdot 2/2 = 4xy^2$ with an energy consumption of $y\cdot (4xy-1)\cdot 0 + y \cdot 2\cdot 1 = 2y$.
	We certainly have enough clause agents for this task: Since $F$ is satisfiable, each clause $c_i$ has at most 
	$\deg(c_i)-1$ unsatisfied literals and thus has enough agents on its clause node to send one agent each to the corresponding paths 
	in which they are needed. 
	As for the timing:  We can easily verify that each needed slow agent can reach its internal path node exactly by the time 
	at which the package arrives at the same node, thus the package does not have to wait for its next agent.
	
	Hence transporting the package from $u_1$ to $u_{x+1}$ takes time $x\cdot 4xy^2 = 4x^2y^2$ and needs an energy of $x\cdot 2y = 2xy$,
	yielding an overall delivery time of $\T= 12x^2y^2+4x^2y^2+16x^2y^2 = 32x^2y^2$ and energy consumption of $\E = 2xy$.
\end{proof}

\begin{lemma}[Delivery schedule yields \psat assignment]
	Any delivery schedule with delivery time $\T=32x^2y^2$ needs an energy consumption of at least $\E\geq 2xy$. 
	Furthermore, if there is a delivery schedule with $\T=32x^2y^2$ and $\E=2xy$, we can retrieve a satisfiable assignment
	to the variables of the underlying 3CNF $F$ from the delivery schedule.
	\label{lem:delivery-sat}
\end{lemma}
\begin{proof}
	We already know by Lemma~\ref{lem:minimum-time} that any delivery schedule with delivery time $\T=32x^2y^2$ can't deliver the
	package via any clause node and can't use the fastest agent on the left side of $u_{x+1}$.
	Thus the length of the package's trajectory from $u_1$ to $u_{x+1}$ must have length $d_G(u_1,u_{x+1}) = 4x^2y^2+xy$.
	We only have slow agents ($\pv=1,\pw=0$) and fast agents ($\pv=2,\pw=1$) available to travel this distance. 
	Hence to cover $4x^2y^2+xy$ in time $\leq 4x^2y^2$ we need to use fast agents for a total length of at least $2xy$. 
	This, however, results in an energy consumption of at least $2xy\cdot 1 = 2xy$ and hence $\E\geq 2xy$.

	To achieve equality, all fast agents must travel a total length of \emph{exactly} $2xy$, moving only while carrying the package, and
	moving \emph{only from left to right}. In particular, fast agents can't help to transport the package over edges of length $4xy-1$, from neither side.

	Therefore, in any delivery schedule with delivery time $\T=32x^2y^2$ and $\E=2xy$, 
	the package must be transported over edges of length 2 by fast agents and over edges of length $4xy-1$ by slow agents. 
	Furthermore, following the same line of reasoning as in the proof of Lemma~\ref{lem:sat-delivery}, to actually achieve $\T=32x^2y^2$ 
	the package can never wait at an internal path node for the agent which will carry it next.
	Suppose the package just arrived at the $2l-1$th node of path $p_{j,\text{true}}$. 
	This means it has covered a distance of $12x^2y^2 + (j-1)\cdot y \cdot(2+4xy-1) + (l-1)\cdot (2+4xy-1)+2$ 
	with the help of $(j-1)\cdot y + l$ fast agents so far. Hence the current time is
	\[	\tfrac{12x^2y^2 + (j-1)\cdot y\cdot(4xy-1) + (l-1)\cdot (4xy-1)}{1} + \tfrac{(j-1)\cdot y\cdot 2 + l\cdot2}{2}
		= 12x^2y^2 + (j-1)4x^2y^2 + (l-1)4xy + 1.  \]
	Among all slow agents, only a slow agent on an adjacent node can reach the internal path node at this time, since all other 
	slow agents are further away. 

	In conclusion, in a delivery schedule with delivery time $\T=32x^2y^2$ and $\E=2xy$, agents stationed on a clause node $c_i$ can only help
	carrying the package on paths corresponding to the literals of this clause, 
	and since there are only $\deg(c_i)-1$ agents stationed on that clause, we know that for at least one literal $l(u_j)$ of $c_i$, 
	the path $p_{j,l(u_j)}$ has been taken (for which no help from a clause agent of $c_i$ is necessary).
	Hence we can read a satisfiable variable assignment for $F$ directly from the choice of the variable agents 
	(which each pick the adjacent \emph{true}- or the adjacent \emph{false}-path): with this assignment, each clause has at least one satisfied literal.
\end{proof}

We remark that the graph $G(F)$ created from $H(F)$ is planar and that all edge lengths and agent velocities and weights have polynomial size. 
Hence we conclude:
\vspace{2ex}

\begin{theorem*}[\ref{thm:fast-efficient-hardness}]
\TWDelivery is \NP-hard, even on planar graphs.	
\end{theorem*}

\subsubsection*{Remark}
The proof of Theorem~\ref{thm:fast-efficient-hardness} also extends to the following relaxation of \TWDelivery:
\textbf{Variant (ii$^*$):} Given a time constraint $\T^*$, minimize the energy $\E$ such that $\T \leq T^*$.

Note that Variant (ii$^*$) includes the setting where $\T^* = \min \T$. 
Hence our hardness result for \TWDelivery extends to Variant (ii$^*$).
Furthermore, in this case $\E$ is NP-hard to approximate to within any constant factor $C$: 
This follows by appropriately scaling the edges traveled by slow agents in the hardness proof for
Theorem~\ref{thm:fast-efficient-hardness}.
Specifically, we scale the edges of length $(4xy-1)$ to $(C\cdot 4xy-1)$ and adapt all edges incident to the 
starting positions of slow agents or incident to $s$, $u$, $t$ accordingly.
Then a \psat-solution will have a delivery time of $\T = C\cdot 32x^2y^2 =: \T^*$ and an energy consumption of $\E = 2xy$,
while any other solution with $\T = C\cdot 32x^2y^2$ will have energy consumption $\E \geq C\cdot 4xy -1$.

\newpage
\section{An efficient algorithm for \TWDelivery on paths}

\subsection{Fast computation and recombination}

\begin{figure}[h!]
	\includegraphics[width=\linewidth]{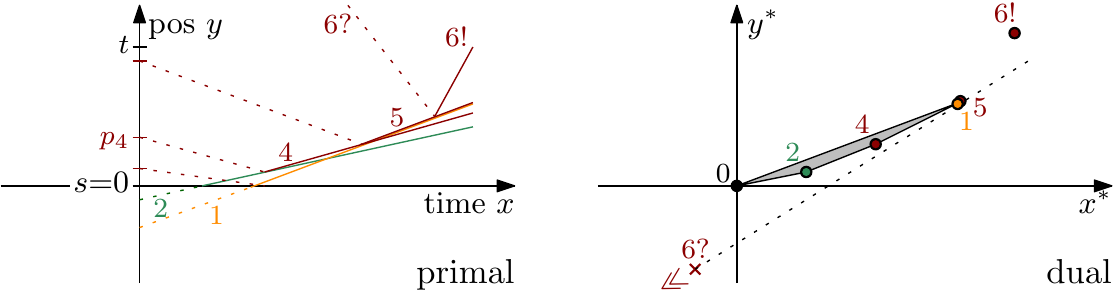}
	\caption[Time first: Line -- dual]{(left) Considering the potential contribution of agent $6$ amounts to
	finding the intersection of a query line (6?) with the upper envelope of previous rays 
	and then (since $\pv_6 > \pv_5 = \pv_1$) adding its corresponding ray (6!).
	(right) The same operations in the geometric point-line dual, 
which allows for faster computations via convex hull queries.}
	\label{fig:appendix:decomposition-dual}
\end{figure}

\paragraph{Envelope intersections via point-line duality.}
In classic geometric point-line duality~\cite{edelsbrunner1987algorithms}, each line $g\colon y=a \cdot x+b$ 
is mapped to a point $g^* = (a, -b)$ and vice versa each point $q = (c,d)$ is mapped to a line $q^*\colon y^* = c \cdot x^* - d$.

In this dual setting, the task of finding the leftmost intersection point of a query line with a set of lines turns into 
finding the line of minimum slope of a query point with a set of points.
This corresponds to asking for the right tangent $q_i^*\colon y^* = x_i \cdot x^* - y_i$ from a query point 
(the dual of the query line) onto the convex hull of a point set 
(the dual of the rays $f_i$\footnote{For the purpose of line-point duality, we treat the rays as lines. 
This is fine, since each ray $f_i \not\equiv 0$ is dominated for $x$-values smaller than $x_i$ anyways.}), 
see Figure~\ref{fig:appendix:decomposition-dual}. 
In order to use such \emph{tangent queries}, we use the data structure for planar convex hulls of dynamic point sets by
Brodal and Jacob~\cite{RikoConvexHull02,jacob2002dynamic}.
To find the corresponding ray of maximum slope $\pv_j$, we also work in the geometric dual: 
Now that we have found the tangent $t^*$ (of direction $(x_i,1)$), we make an \emph{extreme point query} 
in the orthogonal direction $(-1,x_i)$. The query returns an extremal vertex of the convex hull in this direction.
In the (original) geometric primal, this vertex corresponds to a ray containing the intersection point $(x_i,y_i)$. 
In case there is only one such ray, we have found $f_j$ (respectively its dual). 
However, if there are multiple such rays, the resulting extremal point -- being a vertex of the convex hull -- 
corresponds to a ray of \emph{minimal} or \emph{maximal} slope among all such rays. Hence we make an additional \emph{neighbor query}, 
giving us the two neighboring vertices on the convex hull. We check whether the neighbor with higher $x^*$-coordinate lies
on the tangent $t^*$, too. If so, the ray $f_j$ of maximal slope $\pv_j$ is given as the dual of this neighbor, 
otherwise as the dual of the extremal point query's result itself. 

Finally, we check whether $\pv_j \leq \pv_i$, and if so, we insert $f_i^* = (\pv_i, \pv_i\cdot x_i - y_i)$ into the point set 
of the geometric dual. The mentioned convex hull data structure allows point insertions and deletions in $\bigO(\log k)$ amortized time, 
and tangent, extremal point and neighbor queries in time $\bigO(\log k)$. 
As for each $f_i$ we need one of each type of queries and one insertion, we get an overall running time of $\bigO(k\log k)$.

\begin{figure}[t!]
	\includegraphics[width=\linewidth]{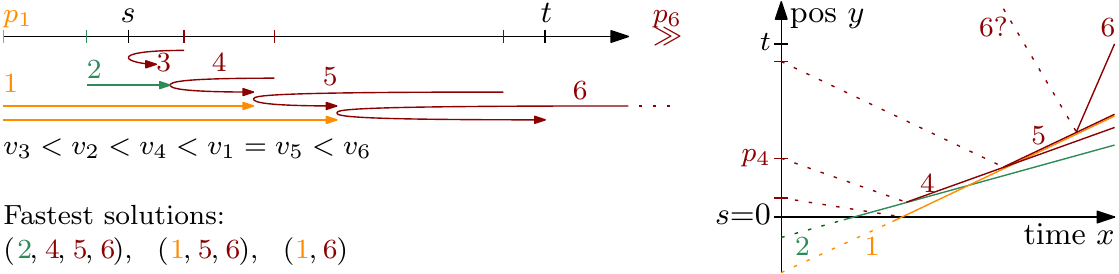}
	\caption[Time first: Line -- decomposition]{(left) Example of possible delivery schedules: 
	If agent $4$ is involved, he must take over the package from agent $2$, since agent $3$ is too slow. 
	Using agents $2$ and $4$ to bring the package to agent 5's pick-up position takes the same time as using agent $1$ on its own. 
	Agents $1$ and $5$ have the same velocity, so in terms of delivery time we could use either or even both of them, 
	but agent $1$ only if agents $2$ and $4$ are both not used (otherwise they all consume energy).
	(right) Fastest solutions correspond to at most $1$ agent with $p_i < s$ and a number of agents
	corresponding to a suffix of the upper envelope.
	}
	\label{fig:appendix:decomposition}
\end{figure}

\paragraph{Uniform velocity instances and recombination.} As any optimum schedule transports the package to $t$ as early as 
possible, the delivery time $\T$ thus corresponds to the $x$-coordinate at which the upper envelope reaches $y=t$. 
Every line segment of the upper envelope corresponds to intervals on the line of piecewise uniform velocity.
Hence we would like to represent the upper envelope as an ordered list of its segments. 
Since we have already computed all the rays $f_i$, computing this list of segments can be done in time $\bigO(k\log k)$: 
This can be done either by a divide-and-conquer approach by Hershberger~\cite{Hershberger89} which finds the upper envelope 
of line segments (where we treat rays as segments extending up to $x=\T$) or with a Graham scan-like sweep~\cite{graham1972efficient}, 
first sorting the rays in order of increasing slopes, then adding the rays one by one in amortized constant time, 
maintaining a stack of lines appearing on the upper envelope~\cite{wcipeg}.

Assuming that we know the optimum schedule for each of the uniform velocity intervals, we can recombine these subschedules as follows:
Recall that either one or no agent with $p_i < s$ is involved in an optimum schedule. In the first case, the optimum schedule simply 
corresponds to the complete upper envelope. In the second case, the optimum corresponds to a ray with starting position on the 
$x$-axis together with a suffix of the upper envelope to which it is a right tangent.\footnote{For example, 
	in Figure~\ref{fig:appendix:decomposition-dual} (right), $f_2$ is a right tangent to $q_4=(x_4,y_4)$, 
	hence we consider the suffix $4,5,6$. On the other hand, $f_1$ is a right tangent to both $q_5=(x_5,y_5)$ and $q_6=(x_6,y_6)$, 
	hence we consider $f_1$ up to $q_5$ with the suffix $5,6$ \emph{as well as} $f_1$ up to $q_6$ with the line segment $6$ only.}
Each ray is a tangent to at most two distinct vertices of the upper envelope, and we treat the line segment defined by 
each such tangent as a uniform velocity instance itself. 
Overall, we get at most $\bigO(k)$ many time-optimal schedules (see Figure~\ref{fig:appendix:decomposition}). 

Out of these time-optimal candidates, it remains to find the schedule with minimum energy consumption $\E$. 
Hence we precompute for each suffix of line segments on the upper envelope the sum of energy consumptions of the corresponding subschedules. 
Given the energy consumption of each subschedule, all suffix-sums can be computed in overall time $\bigO(k)$. 
After this precomputation, the schedule lexicographically minimizing $(\T,\E)$ is given as the minimum over all candidates. 
We conclude:

\begin{lemma*}[\ref{lem:decomposition}]
Arbitrary velocity instances of \TWDelivery on paths can be decomposed into and recombined from 
uniform velocity instances in time $\bigO(n+k\log k)$.
\end{lemma*}

\subsection{Characterization of handover points}

\begin{lemma*}[\ref{lem:line-structure}]
There is an optimum solution where each agent $i$ that is involved in advancing the package
picks it up at $q_i^+ = q_i$ or at $q_i^+ = p_i$.	
\end{lemma*}

\begin{proof}
	Assume otherwise, i.e.,~ all optimum schedules contain some point $q_i^+ \notin \left\{ q_i,p_i \right\}$. 
	For the sake of a contradiction let $S$ be the optimal schedule where the coordinate of the first pick-up point $q_i^+$ 
	different from $q_i$ and $p_i$ is maximized, and let $\epsilon>0$ be an arbitrary small constant.
	Denote by $j$ the agent transporting the package right before agent $i$.
	Note that $j$ always exists as we must have $i > 1$ by $q_1^+ = s = q_1$. 
	We know that $q_i < q_i^+$ as $i$ would delay the delivery otherwise.
	We also know that $q_i^+ < p_i$ as the transportation on $[p_i, p_i+\epsilon]$ would otherwise use more energy than necessary 
	($\pw_j > \pw_i$). We now distinguish three cases:
	\begin{itemize}
		\item	$2\pw_i < \pw_j$: We modify $S$ by moving $q_i^+$ by $\epsilon$ to the left towards $q_i$.
			This does not change the delivery time but reduces the energy cost and thus contradicts $S$'s optimality.
		\item	$2\pw_i > \pw_j$: If we modify $S$ by moving $q_i^+$ by $\epsilon$ to the right, the energy cost goes down, 
			again contradicting $S$ being optimal.
		\item	$2\pw_i = \pw_j$: If we again modify $S$ by moving $q_i^+$ by $\epsilon$ to the right, 
			the delivery time and energy cost do not change, but we contradict the extremality of $q_i^+$ in $S$.\qedhere
	\end{itemize}
\end{proof}

\subsection{Details of the dynamic program.}

The computational bottleneck of our dynamic program is (for each subproblem $\E[i]$)
the minimization over the set of options in Case~2.c). Each option evaluates a linear function 
\begin{align*}
	f_j(q_i):= \underbrace{\pw_j}_{\mathclap{\text{slope } a}} \cdot \ q_i\ + \ \underbrace{(\E[j] - p_j \cdot \pw_j)}_{\text{offset } b} 
		= a \cdot q_i + b
\end{align*}
at position $q_i$. So when we see $\min_j \{f_j(q_i) \mid q_i < p_j < p_i \text{ and } A[j]=j \}$ as a function of $q_i$, 
we get the lower envelope of all these linear functions (see Figure~\ref{fig:time:primal-dual} (left)).
The set of these functions is fully dynamic over time: Whenever we compute $\E[i]$ with $A[i]=i$, a new function $f_i$ is added to the set.
Whenever we go from $q_i$ to $q_{i+1}$, all functions $f_j$ with $q_i < p_j < q_{i+1}$ drop out of the window of functions 
we are minimizing over, which we can take care of by lazy incrementation. 
One particular data structure that supports insertion and deletion of up to $k$ linear functions $ax+b$ as well as 
find-min-queries for up to $k$ query values $x=q$ is a parametric heap. If successive query values $q$ are non-decreasing, 
a parametrized data structure is also called kinetic~\cite{Basch99}. 
The running time of a kinetic heap, however, is dominated by delete-operations, 
which take $\bigO(\log k \log\log k)$ time each~\cite{Kaplan01}.
Hence we turn once more to the more recent dynamic planar convex hull data structure, 
which supports deletions in amortized time $\bigO(\log k)$~\cite{RikoConvexHull02}. 
Overall we get:

\begin{figure}
	\centering\includegraphics[width=\linewidth]{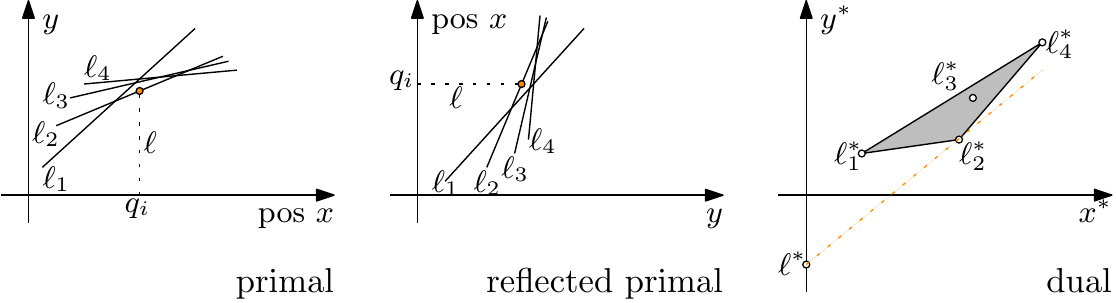}
	\caption[Time first: Line -- duality]{Use point-line duality to transform the minimization problem of
		(left) finding the intersection point of a query line with a lower envelope of lines into
		(right) a tangent query problem on the convex hull of a dynamic point set.
	}
	\label{fig:time:primal-dual}
\end{figure}

\begin{lemma*}[\ref{lem:dynamic-dp}]
An optimum schedule for \TWDelivery with uniform velocity $\overline{\pv}$ on the line can be computed in $\bigO(n + k \log k)$ time.
\end{lemma*}

\begin{proof}
	We formulate our optimization problem for agent $i$ in Case $(2.c)$ as the following equivalent geometric problem:
	Given up to $k$ linear equations $\ell_j$: $y=a_j \cdot x+b_j$ with $a_j > 0$\footnote{Note that it is without loss of generality 
	that we assume $a_j = \pw_j \neq 0$ here, as we have $0 \leq \pw_i < \pw_j$.} 
	and a vertical query line $\ell$: $x=q_i$, find the lowest intersection point between $\ell$ and any of the $\ell_j$ 
	(see Figure~\ref{fig:time:primal-dual} (left)). As before, we want to apply point-line duality, 
	where each line $y=a \cdot x+b$ is mapped to a point $(a, -b)$ and vice versa.
	However, as the vertical line $\ell$ does not have a (finite) dual representation, 
	we first swap the coordinates, reflecting the whole configuration at axis $y=x$ (see Figure~\ref{fig:time:primal-dual} (center)).
	We get $\ell_i$ as $x = \frac{y}{a_i}-\frac{b_i}{a_i}$ and its dual 
	$\ell^*_i$ as $\smash{\left(\frac{1}{a_i}, \frac{b_i}{a_i}\right)}$ and can now also map 
	$\ell$: $x=q_i$ to $\ell^*$: $(0,-q_i)$ (see Figure~\ref{fig:time:primal-dual} (right)).
	
	In this dual setting, the task is to find the line of smallest slope (= leftmost point in the mirrored primal) 
	that goes through $\ell^*$ and one of the $\ell^*_j$. This corresponds to asking for the right tangent through $\ell^*$ 
	onto the convex hull of the point set $\{\ell_j' \mid q_i < p_j < p_i \text{ and } A[j]=j \}$. 
	We can find $\ell_j^*$ itself with an extreme point query to get $A'[i] = j$.
 	Each operation takes (amortized) time $\bigO(\log k)$, and as we perform at most $4k$ operations overall, 
	we get a total running time of $\bigO(k \log k)$.
\end{proof}

\section{NP-hardness of \CDelivery}

\NP-hardness of \TWDelivery followed from instances where $\T$ and $\E$ were lower bounded, 
i.e. $\T \geq \T^*$ for some $\T^*$ and for every schedule with $\T=\T^*$, $\E \geq \E^*(\T^*)$ for some $\E^* = \E^*(\T^*)$ depending on $\T^*$.
Finding a delivery schedule with time and energy $(\T^*,\E^*)$ turned out to be equivalent to solving a corresponding instance of \psat.
By scaling the weights of all agents by a factor $\delta$, this turns into finding a delivery schedule with time and energy $(\T^*,\delta\E^*)$.

Set $\delta := \epsilon/8$. Consider any delivery schedule with time/energy $(\T,\E) \neq (\T^*, \delta\E^*)$. 
Either we have $\E > \delta\E^*$ and thus $\epsilon\T+(1-\epsilon)\E > \epsilon\T^*+(1-\epsilon)\delta\E^*$ immediately.
Or we have $\E < \delta\E^*$. 
In this case, we let the fast agents of velocity $2$ and weight $\delta$ cover less than the original distance $2xy$ in Lemma~\ref{lem:delivery-sat}.
This saves an energy amount of $z\cdot \delta$ for some $z, 0<z\leq 2xy$. 
However, we have to replace this distance by either a slow agent, or by letting the fastest agent of velocity $8$ do more work. 
Either way, we increase the delivery time by at least $z/8$, and we again get
\[
	\epsilon\T+(1-\epsilon)\E	 \geq \epsilon(\T^*+z/8) + (1-\epsilon)(\delta\E^*-z\delta) \\
 = \epsilon\T^*+(1-\epsilon)\delta\E^* + \underbrace{z(\epsilon/8 - (1-\epsilon)\delta)}_{>0}. \]
Thus, finding a delivery schedule of minimum $\epsilon\T+(1-\epsilon)\E$ yields a delivery schedule 
with delivery time and energy consumption $(\T^*,\delta\E^*)$ from which we can infer a satisfiable assignment as in Lemma~\ref{lem:delivery-sat}.
We conclude:
\vspace{2ex}

\begin{theorem*}[\ref{thm:combination-hardness}]
\CDelivery is \NP-hard for all $\epsilon\in(0,1)$, even on planar graphs.
\end{theorem*}

\end{document}